\documentclass[onecolumn,10pt]{IEEEtran}
\usepackage{flushend}
\usepackage{cite}
\usepackage{color}
\usepackage{epsf}
\usepackage{epsfig}
\usepackage{graphicx}
\usepackage{graphics}
\usepackage{epstopdf}
\usepackage[small]{caption}
\usepackage{amsmath}
\usepackage{amssymb}
\usepackage{amsthm}
\usepackage{amsxtra}
\usepackage[ruled]{algorithm2e}
\usepackage{algorithmic}
\usepackage{enumerate}
\usepackage{multirow}
\usepackage{enumerate}
\usepackage{amssymb}
\usepackage{lipsum}
\usepackage{setspace}
\usepackage{multirow}
\usepackage{wasysym}
\usepackage{enumerate}
\usepackage{amssymb}
\usepackage{lipsum}
\usepackage{setspace}
\usepackage{multirow}
\usepackage{subcaption}
\usepackage{array}
\usepackage{makecell}
\usepackage{geometry}
\usepackage{colortbl}
\usepackage{color,soul}
\newtheorem{theorem}{\bf Theorem}
\newtheorem{proposition}{\bf Proposition}

\newtheorem{definition}{\bf Definition}
\newtheorem{remark}{Remark}

\usepackage{rotating} 
\usepackage{graphicx} 

\newcommand{\Rmnum}[1]{\expandafter\@slowromancap\romannumeral #1@}

\makeatother
\begin{document}
\title{A Coalition Formation Game Framework for Peer-to-Peer Energy Trading}
\author{Wayes Tushar$^{a,*}$, Tapan Kumar Saha$^a$, Chau Yuen$^b$, M. Imran Azim$^a$, Thomas Morstyn$^c$, H. Vincent Poor$^d$, Dustin Niyato$^e$, and Richard Bean$^a$\\
$^a$The University of Queensland, Brisbane, Australia\\
$^b$Singapore University of Technology and Design, Singapore\\
$^c$The University of Oxford, United Kingdom\\
$^d$Princeton University, NJ, USA\\
$^e$Nanyang Technological University, Singapore
\thanks{$^*$Corresponding author at: School of Information Technology and Electrical Engineering, The University of Queensland, Brisbane, QLD 4072, Australia.}
\thanks{E-main addresses: w.tushar@uq.edu.au (W. Tushar), saha@itee.uq.edu.au (T. K. Saha), yuenchau@sutd.edu.sg (C. Yuen), m.azim@uq.net.au (M. I. Azim), thomas.morstyn@eng.ox.ac.uk (T. Morstyn), poor@princeton.edu (H. V. Poor), dniyato@ntu.edu.sg (D. Niyato), r.bean1@uq.edu.au (R. Bean).}
}
\IEEEoverridecommandlockouts
\maketitle
\doublespace
\begin{abstract}
This paper studies social cooperation backed peer-to-peer energy trading technique by which prosumers can  decide how they can use their batteries opportunistically for participating in the peer-to-peer trading. The objective is to achieve a solution in which the ultimate beneficiaries are the prosumers, i.e., a prosumer-centric solution. To do so, a coalition formation game is designed, which enables a prosumer to compare its benefit of participating in the peer-to-peer trading with and without using its battery and thus, allows the prosumer to form suitable social coalition groups with other similar prosumers in the network for conducting peer-to-peer trading. The properties of the formed coalitions are studied, and it is shown that 1) the coalition structure that stems from the social cooperation between participating prosumers at each time slot is both stable and optimal, and 2) the outcomes of the proposed peer-to-peer trading scheme is prosumer-centric. Case studies are conducted based on real household energy usage and solar generation data to highlight how the proposed scheme can benefit prosumers through exhibiting prosumer-centric properties.
\end{abstract}
\section*{Nomenclature}
\addcontentsline{toc}{section}{Nomenclature}
\begin{IEEEdescription}[\IEEEsetlabelwidth{$V_1,~V_2,$}]
\item[$\mathcal{A,~B,~S}$]~Different coalitions.
\item[$b_n$]~Capacity of the battery of prosumer $n$.
\item[$b_{n,c}$]~Available capacity of the battery of prosumer $n$ at $t$.
\item[$\mathbb{D}_{hp}, \mathbb{D}_{c}$]~Symbols for stability of coalition.
\item[$E_{n,\text{pv}}(t)$]~Solar generation of prosumer $n$ at $t$.
\item[$E_{n,\text{dis}}(t)$]~Energy discharged from the battery of prosumer $n$ at $t$.
\item[$E_{n,\text{cha}}(t)$]~Energy charged to the battery of prosumer $n$ at $t$.
\item[$E_{n,\text{dem}}(t)$]~Total energy demand of prosumer $n$ at $t$.
\item[$E_{n,\text{hou}}(t)$]~Energy demand of prosumer $n$ for household activities at $t$.
\item[$E_{n,c}(t)$]~Own solar and battery energy consumption by prosumer $n$ at $t$.
\item[$E_{n,\text{def}}(t)$]~Energy deficiency of prosumer $n$ at $t$.
\item[$E_{n,\text{sur}}(t)$]~Energy surplus of prosumer $n$ at $t$.
\item[$E_{\text{chr,dis}}$]~Recommended maximum charging and discharging rate of the battery.
\item[$E_{\text{transfer}}(t)$]~Limit imposed on maximum energy transfer over the network.
\item[$E_c(t)$]~Total energy charged by prosumers at state 2 at $t$.
\item[$E_d(t)$]~Total energy discharged by prosumers at state 2 at $t$.
\item[$e_\text{max}(t)$]~$\max(E_{\text{chr,dis}},E_{\text{transfer}}(t))$.
\item[$e_{n,c}(t)$]~Energy charged by prosumer $n$'s battery at $t$.
\item[$e_{n,d}(t)$]~Energy discharged by prosumer $n$'s battery at $t$.
\item[$J_{n,b}(t)$]~Cost to prosumer $n$ at state 1 at $t$.
\item[$k$]~Scaling factor.
\item[$\mathcal{K}$]~Set of providers in coalition 2.
\item[$K$]~Total number of providers in coalition 2.
\item[$\mathcal{L}$]~Set of receivers in coalition 2.
\item[$L$]~Total number of receivers in coalition 2.
\item[$n,m$]~Index of prosumer.
\item[$N$]~Total number of prosumers.
\item[$\mathcal{N}$]~Set of all prosumers.
\item[$p_s(t)$]~Mid-market selling price per unit of energy at $t$.
\item[$p_b(t)$]~Mid-market buying price per unit of energy at $t$.
\item[$p_c(t)$]~Price per unit of energy charged by the battery at $t$.
\item[$p_d(t)$]~Price per unit of energy charged by the battery at $t$.
\item[$p_l$]~Degradation cost per kWh of the battery.
\item[$p_{n,c}(t)$]~Threshold price of prosumer $n$ to charge its battery at $t$.
\item[$p_{n,d}(t)$]~Threshold price of prosumer $n$ to discharge its battery at $t$.
\item[$p_{d,s2}(t)$]~P2P discharging price for prosumers in state 2 at $t$.
\item[$p_{c,s2}(t)$]~P2P charging price for prosumers in state 2 at $t$.
\item[$s_n(t)$]~State of charge of $n$'s battery at $t$.
\item[$t$]~Index of each time slot.
\item[$U_{n,s}(t)$]~Utility of prosumer $n$ at state 1 at $t$.
\item[$U_{n,c}(t)$]~Utility of prosumer $n$ at state 2 for charging its battery at $t$.
\item[$U_{n,d}(t)$]~Utility of prosumer $n$ at state 2 for discharging its battery at $t$.
\item[$U_\mathcal{\cdot}(n)$]~Utility of $n$ as a part of coalition $\mathcal{\cdot}$.
\item[$\mathcal{V}$]~Set of prosumers at state $2$.
\item[$\mathcal{W}$]~Set of prosumers at state $1$.
\item[$\alpha_n$]~Satisfaction parameter of prosumer $n$.
\item[$\beta$]~Pricing parameter.
\item[$\Gamma$]~Symbol for coalition formation game.
\item[$\nu$]~Value of coalition.
\end{IEEEdescription}
%
\section{Introduction}\label{sec:introduction}
The extensive growth in distributed energy resources in recent times is not only able to supply the growing energy demand of the consumers, but can also facilitate a notable mix of clean renewable energy into the grid \cite{Sousa_RSER_Apr_2019}. For this reason, it is critical to enable extensive participation of owners of these assets in the deregulated energy market to provide frequency control services, demand response, profit maximization, operation reliability, planning, combating uncertainty in generation, and managing networked distribution system~\cite{Peck_Spectrum_2017}.  As such, to enable prosumers participation, the Feed-in-Tariff (FiT) scheme has been in the market since the last decade\cite{Liang-ChengYe-AE:2017}. However, the benefit to prosumers for participating in the FiT has been proven to be very marginal~\cite{Tushar_SPM_2018}. Consequently, a number of FiT schemes has been discontinued in recent times \cite{FiTPolicy_2015}. Meanwhile, a new energy trading paradigm, known as peer-to-peer (P2P) energy trading, has emerged recently \cite{Jogunola_Energies_Dec_2017}.

P2P trading is an emerging economic model that allows commodities such as energy to go from one prosumer to another rather than from the grid to the consumer \cite{Morstyn_NE_2018,Tushar_SPM_2018}. This trading platform allows a prosumer to take advantage of other prosumers within its community who produce or have more energy than they need by buying energy from them at a relatively cheaper rate and vice versa. Due to its potential in revolutionizing the energy domain, several studies are conducted that focus on the financial model of P2P trading and the impact of network constraints on P2P. In developing the financing model, the authors in \cite{Si_AE_Dec_2018} investigate the economic impact of the energy complementary strategy to promote the sustainable development of the urban energy system. An optimization model is developed in \cite{Nguyen_AE_Oct_2018} for the photovoltaic (PV)-battery system in the P2P market with real-world constraints and market signals. A similar PV and battery-driven P2P energy trading model is also proposed in \cite{ChaoLong_AE_2018} through an aggregated two-stage battery control technology. Trading mechanisms that focus on empowering prosumers in the market are proposed in \cite{Melendez_AE_Aug_2019} and \cite{Paark_Sustainability_Mar_2018}. In the literature, integration of P2P trading in the energy market is also discussed via double-auction based~\cite{Al-Baz_AE_May_2019}, fairness based~\cite{Moret_TPWRS_EA_2018}, consensus-based~\cite{Sorin_TPWRS_Mar_2019}, negotiation-based~\cite{Moret_PSCC_June_2018}, generalized Nash equilibrium~\cite{Cadre_Elsevier_Press_2019}, and orchestrator based~\cite{Siozios_DATE_Mar_2019} approaches respectively.

In terms of investigating the impact of the physical network on P2P trading, \cite{Chapman_TSG_2018} and \cite{Guerrero_AUPEC_Nov_2017} discuss how P2P trading can potentially jeopardize the reliability of the power system by increasing bus voltages and propose potential solutions to avoid such circumstances. A Jacobi-proximal alternating direction method of multipliers (ADMM) is applied in \cite{Hamada_AE_Apr_2019} to design the P2P grid voltage support function for smart PV inverters with the purpose of regulating the voltages within a reasonable time. The impact of the level of transmission losses on P2P trading behaviors of retailers and consumers is discussed in \cite{Zhang_AE_Feb_2019} using a credit rating based multi-leader multi-follower game.  In \cite{Werth_TSG_2018}, the authors use the information and communications technologies (ICT) concept of network overlays and P2P networks to improve the resilience of the network against utility blackouts. ICT for local smart grids is also utilized in \cite{Pouttu_Conference_EUCNC_2017} to design a P2P model for distributed energy trading and grid control. Further, a number of pilot projects such as Brooklyn Microgrid in the USA \cite{Mengelkamp_AppliedEnergy_2017}, Valley Housing Project in Fremantle, Western Australia \cite{PowerLedger_2017}, and the P2P-SmartTest Project in Europe~\cite{SmartTest_AE_2018} are under development to demonstrate the benefit of P2P trading.

While the benefits of P2P trading are clear, the key concern that yet to be addressed is that how to develop a P2P trading mechanism, which is capable of ensuring a continuous and sustainable operation of the energy trading between prosumers. Finding a favourable answer to this question is particularly important due to following reasons: 1)  P2P trading emphasizes on transfer of energy among multiple prosumers with minimum (if not at all) influence and control from a central controller, e.g., the retailer. This establishes P2P as a trustless system \cite{PowerLedger_2017}. Therefore, encouraging prosumers to cooperate with one another for trading energy could be challenging; 2) Energy trading involving active prosumers participation is not a new phenomenon. However, a sustainable and continual operation of such a scheme cannot be guaranteed unless prosumers extensively participate in the trading. Therefore, it is utmost necessary that the developed energy trading schemes are \emph{prosumer-centric}, in which the main receipient of the benefits are the prosumers as both the buyer and the sellers of energy~\cite{Tushar-TSG:2014}. Otherwise, the P2P trading could potentially be discontinued like many other existing schemes, e.g., see \cite{FiTPolicy_2015} and \cite{Colley_Misc_2014}.

Given this context, we propose a P2P energy trading scheme to address these issues by exploiting \emph{social cooperation} among different prosumers. In particular, we use the framework of a cooperative game\cite{Thomas_PES_2018} to capture the interaction between different prosumers. However, unlike \cite{Thomas_PES_2018}, where a canonical coalition game is modeled, we use the framework of a coalition formation game to determine a stable network structure, in which each prosumer can decide whether it wants to be in state $1$ or state $2$ at a particular time slot for cooperating with other similar prosumers to participate in the P2P trading. State $1$  refers to the state of a prosumer, in which it does not charge or discharge its battery for participating in P2P trading. It only uses its surplus generation for participating in the market. State $2$ of a prosumer, on the other hand, refers to the state, in which it is willing to charge or discharge its battery for buying energy from or selling energy to the P2P market. The use of the proposed coalition formation game can effectively help prosumers to decide with whom to cooperate or not, which helps the designed energy trading model to achieve a \emph{prosumer-centric} solution. 

To this end, the main contributions of the paper are: 
\begin{itemize}
\item We propose a P2P energy trading scheme between different prosumers of an energy network by using a coalition formation game. The objective is to exploit the social cooperation between the prosumers to attain a prosumer-centric solution.
\item We develop an algorithm that helps a prosumer to select a state based on which it can choose whether or not it is beneficial to put its battery in the P2P market and subsequently form a suitable group based on the coalition game. Within the selected state, the algorithm also enables the prosumer to decide whether to charge or discharge its battery to participate in the proposed P2P trading.
\item We study the properties of the coalition formation game and show that the formed coalition among prosumers for P2P trading is stable and socially optimal that delivers prosumer-centric outcomes; and
\item We validate the properties of the proposed scheme through numerical simulation based on real consumers data.
\end{itemize}

While in a number of existing literature such as \cite{Wong_JSAC_July_2012}, \cite{Shams_Energy_July_2018}, and \cite{Baringo_TPWRS_May_2019}, prosumers make their decision of energy trading price through a day ahead scheduling, this study develops a P2P scheme that considers the situation at the current time slot to model the decision making process of the prosumers.  Such a model is particularly beneficial to apply in scenarios where the history of PV generation patterns and electricity demand of prosumers is not available, for example, communities where P2P trading is deployed as a smart energy management scheme for the first time. Note that example of such current time slot based energy management scheme can also be found in \cite{BoChai_TSG_2014,Maharjan_TSG_Mar_2013,Tushar_TSG_Press_2019} and \cite{Tushar_TSG_2_2016}. Further, we stress that prosumer-centric solutions have also been discussed in \cite{Morstyn_PWRS_2018} and \cite{Tushar_Access_Oct_2018}. However, the proposed study is different from them in terms of considered system model, game formulation, and analysis.

The rest of the paper is organized as follows. We introduce the system model for the proposed P2P trading in Section~\ref{sec:systemmodel}. We propose a P2P energy trading model based on a coalition formation game in Section~\ref{sec:GameFormulation}. Properties of the studied P2P trading are studied in Section~\ref{sec:properties} followed by numerical simulation results in Section~\ref{sec:casestudy}. Finally, we conclude the paper in Section~\ref{sec:conclusion}.

\section{System Model}\label{sec:systemmodel} 
We consider a P2P energy network consisting of a centralized power system (CPS) and $N$ houses, where $N=|\mathcal{N}|$. Each house $n\in\mathcal{N}$ acts as a prosumer, which is equipped with a solar panel and a battery of capacity $b_n$. Prosumer $n$ can use the energy from its solar panel, battery, or from the grid and other entities to meet its demand. Alternatively, it can transfer its surplus energy, if there is any, to the grid, to its own battery or to other prosumers. To do so, each prosumer is also equipped with a smart meter, which can 1) determine and record the generation of energy from the rooftop solar panel, 2) determine and record the consumed energy by the prosumer, 3) determine and record the state of charge (SoC) of the battery, 4) manage the charging and discharging of a battery, and 5) determine and record the energy that a prosumer $n$ sells to or buys from the CPS, the battery, or from another entity, when necessary. 

The considered P2P network is assumed to have two layers~\cite{Mengelkamp_AppliedEnergy_2017}: a physical layer and a virtual layer. The physical layer is responsible for the physical connection and transfer of energy between different energy prosumers within a network through a  distribution system (built and managed by an independent system operator). On the other hand, all prosumers can communicate, exchange information with one another, and decide on their traded energy amount and transaction price over the virtual layer, e.g., blockchain~\cite{PowerLedger_2017} and Elecbay~\cite{SmartTest_AE_2018}. The study presented in this paper focuses on the application of energy trading in the virtual layer of the P2P network.

We assume that at any time slot $t$ of the day,  the generation of solar energy from prosumer $n$'s solar panel is $E_{n,\text{pv}}(t)$, the amount of energy discharged from its battery is $E_{n,\text{dis}}(t)$, and the energy demand of prosumer $n$ is $E_{n,\text{dem}}(t) = E_{n,\text{hou}}(t)+E_{n,\text{cha}}(t)$. Here, $E_{n,\text{hou}}(t)$ is the energy that prosumer $n$ needs for its household activities and $E_{n,\text{cha}}(t)$ is the energy to charge the battery. Since the energy from the solar is free and from the battery is reasonably cheap\footnote{However, there is an associated battery degradation cost~\cite{Peterson_JPS_Apr_2018}.}, in reality the owner of the house prefers to use the solar and battery energy to meet the demand. Therefore, the amount of solar and battery self consumption by each house is
\begin{equation}
E_{n,c}(t) = \min\left(E_{n,\text{pv}}(t)+E_{n,\text{dis}}(t), E_{n,\text{dem}}(t)\right).
\label{Eqn:Equation1}
\end{equation}
Accordingly, based on the generation, demand, and consumption of energy, each prosumer $n$ may need to buy excess energy to meet its energy deficiency $E_{n,\text{def}}(t)$ or has surplus energy $E_{n,\text{sur}}(t)$ to sell. Accordingly,
\begin{equation}
E_{n,\text{sur}}(t) = \left(E_{n,\text{pv}}(t)+E_{n,\text{dis}}(t)\right) - \left(E_{n,\text{hou}}(t) + E_{n,\text{cha}}(t)\right),
\label{Eqn:Equation2}
\end{equation}
and
\begin{equation}
E_{n,\text{def}}(t) = \left(E_{n,\text{hou}}(t) + E_{n,\text{cha}}(t)\right) - \left(E_{n,\text{pv}}(t)+E_{n,\text{dis}}(t)\right).
\label{Eqn:Equation3}
\end{equation}
A prosumer $n$ can sell $E_{n,\text{sur}}(t)$ either to the CPS or to other prosumers in set $\mathcal{N}\setminus\{n\}$ of the network and can buy $E_{n,\text{def}}(t)$ from either of them. Note that in \eqref{Eqn:Equation2} and \eqref{Eqn:Equation3} the surplus and deficiency are calculated respectively after a prosumer already charges it battery from its own generation and discharges the battery for its household demand.

We further consider that the prosumers participating in the P2P energy trading choose to be either of the two states at any given time slot $t$. 
\paragraph{State~1} Prosumers in state $1$ do not charge or discharge their batteries for trading purposes and only sell their surplus $E_{n,\text{sur}}(t)$ to other prosumers in $\mathcal{N}\setminus\{n\}$ or to the grid when $E_{n,\text{sur}}(t) = \left(E_{n,\text{hou}}(t) + E_{n,\text{cha}}(t)\right) - E_{n,\text{pv}}(t)>0$. In other words, a prosumer in state $1$, which we denote as an element of set $\mathcal{W}$, does not put its battery in the energy trading market. Similarly, a state $1$ prosumer only buys energy, when necessary, to meet the deficient energy $E_{n,\text{def}}(t) = \left(E_{n,\text{pv}}(t)+E_{n,\text{dis}}(t)\right) - E_{n,\text{hou}}(t)$. Note that a prosumer without any battery always remains in state $1$.
\paragraph{State~2} Prosumers is state $2$ use their batteries for trading purposes. That is, a prosumer in state $2$, which we indicate via $\mathcal{V}$ is motivated, e.g., for more economical benefit, to discharge its battery for selling energy to other prosumers. Similarly, it may also buy energy from other peers to charge its battery for future needs. 

Clearly, the benefit, as we will see in the next section,  that a prosumer may reap for being a state $1$ or state $2$ during P2P trading could be different due to the additional benefit and cost of using a battery for trading. In fact, such differences influence a prosumer to dynamically choose its states at different time slots of P2P trading. To clarify this further, let us consider the toy example shown in Table~\ref{table:ToyExample}. 

In the table, we focus on how the state of six prosumers can be decided for two time slots. The P2P trading price for charging/discharging battery and the battery degradation price are also assumed to be similar for all of them. Now, due to different demand of energy and subsequent state of charge (SoC) in the battery, the utilities (calculated based on the discussion in the next section) for being in state 1 and state 2 are different for different prosumers. For example, at time slot $1$, both prosumer 1 and 2 are interested to charge their batteries (to be in state 2). However, the positive utility obtained from such charging enable  prosumer 1 to be in state 2, whereas prosumer 1 needs to be in state 1 to avoid the lower utility (which is negative) for becoming a state $2$ prosumer. However, they change their respective states based on the new preferences of charging and discharging in time slot $2$. Similar utility based choice is also shown form other prosumers. Note that when a prosumer does not have a battery (prosumer 6) or it does not want to charge or discharge despite having battery (prosumer 5 in time slot $1$) will remain in state 1.

Clearly, the utility that a prosumer obtain from its participation in the P2P trading influences its choice of different states. As such, we discuss different utility functions that mimic the benefits to each prosumer when they are in state 1 and state 2 in the following section.
\begin{table}[t!]
\centering
\caption{A toy example to demonstrate how prosumers may choose different states for P2P trading. The degradation price is taken from \cite{Peterson_JPS_Apr_2018} and P2P price is calculated following the mid-market rule~\cite{Tushar_Access_Oct_2018} by assuming a grid price of 26 cents per kWh~\cite{Queensland_Electricity_Bill_2019} and a FiT price of 10 cents per kWh~\cite{Queensland_FIT_Bill_2019}.}
\includegraphics[width=0.8\columnwidth]{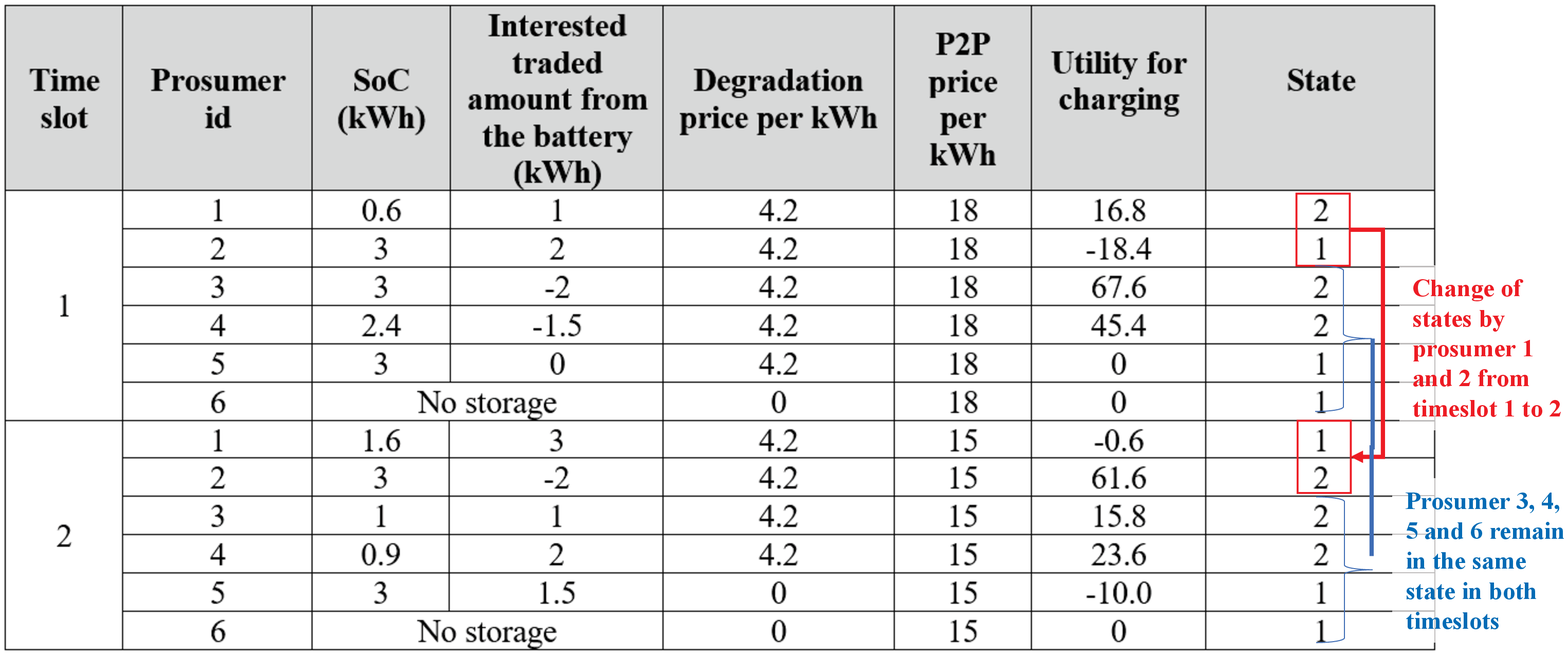}
\label{table:ToyExample}
\end{table}
\subsubsection{Utility of prosumer at state 1}For prosumer at state 1 the contribution of a prosumer's battery to the energy trading can be assumed to be $0$ when considering the flow of battery energy across different entities. Hence, from the perspective of P2P trading, each prosumer $n$ at state $1$ can be considered as a prosumer with the grid-tie solar system without battery~\cite{Tushar_TSG_2017}. Therefore, prosumers of this type trade their solar energy either with one another prosumer or with the grid without charging or discharging their batteries. In this context, if the mid-market buying and selling prices per unit of energy are $p_b(t)$ and $p_s(t)$ respectively (following \cite{Tushar_Access_Oct_2018}), the utility and cost to a prosumer $n$ can be defined as
\begin{equation}
U_{n,s}(t) = p_s(t) E_{n,\text{sur}}(t),~~\text{and}~J_{n,b}(t) = p_b(t) E_{n,\text{def}}(t)
\label{eqn:Utility_Cost_Type1}
\end{equation}
respectively for selling and buying their surplus and deficient energy.
\subsubsection{Utility of a prosumer at state 2}We note that the utility that a prosumer obtains from charging its battery can be defined as the sum of the relevant \emph{cost} and \emph{the benefit} of charging. Benefit of charging can be captured by using  a utility function with decreasing marginal benefit~\cite{Wayes-J-TSG:2012,Samadi_SmartGridComm_2010,BoChai_TSG_2014}. Such a property refers to the fact that the benefit to a prosumer for charging its battery eventually gets saturated with charging as the SoC $s_n(t)$ of the battery becomes close to its available capacity $b_{n,c}(t) = b_n - s_n(t)$. 

To this end, we consider the utility of a state 2 prosumer $n$ for charging its battery as:
\begin{equation}
U_{n,c}(t) = k[b_{n,c}(t)e_{n,c}(t) - \frac{1}{2}\alpha_n e_{n,c}(t)^2] - (p_{l}+p_{c}(t))e_{n,c}(t).
\label{Eqn:Utility_Battery_Charge}
\end{equation}
In \eqref{Eqn:Utility_Battery_Charge}, the term $(b_{n,c}(t) e_{n,c}(t) - \frac{1}{2}\alpha_n e_{n,c}(t)^2$ refers to the benefit of charging, $0<k\leq 1$ is a scaling factor, and $(p_{l}+p_{c}(t))e_{n,c}(t)$ is the cost of charging. Here, $p_l$ is the degradation cost per kWh of the battery~\cite{Ma_TCST_2015}, $p_{c}(t)$ is the price per unit of charging energy $e_{n,c}(t)\leq E_\text{chr,dis}$ (where, $E_\text{chr,dis}$ is the rated charging/discharging rate of the battery), and $\alpha_n>0$ is the satisfaction parameter of prosumer $n$~\cite{Wayes-J-TSG:2012}. 

In the similar way, the utility of a state 2 prosumer $n$ for discharging its battery can be defined as
\begin{equation}
U_{n,d}(t) = k[s_n(t) e_{n,d}(t) - \frac{1}{2}\alpha_n e_{n,d}(t)^2] + (p_{d}(t)-p_{l})e_{n,d}(t),
\label{Eqn:Utility_Battery_DisCharge}
\end{equation}
where $p_{d}(t)$ is the price per unit of discharged energy $e_{n,d}(t)\leq E_\text{chr,dis}$. However, unlike \eqref{Eqn:Utility_Battery_Charge}, $(p_{d}(t)-p_{l})e_{n,d}(t)$ in \eqref{Eqn:Utility_Battery_DisCharge} refers to the revenue that an $n$ obtains from discharging its battery. Further, 
\begin{equation}
e_{n,c}(t) \leq \min(e_\text{max}(t),b_{n,c}(t),E_{n,\text{sur}}(t)),~~~e_{n,d}(t) \leq \min(e_\text{max}(t),s_n,E_{n,\text{def}}(t))
\label{eqn:char-discha-Grp2}
\end{equation}
to ensure that battery charges or discharges neither more than its rated charging and discharging rate nor more than the network limit on the energy transfer. In other words, 
\begin{equation}
e_\text{max} = \max\left(E_\text{chr,dis}, E_\text{transfer}(t)\right), 
\label{eqn:emax}
\end{equation} 
where, $E_\text{transfer}(t)$ is network controller imposed maximum limit of energy transfer over the network at $t$ such that the node voltage does not violate the recommended limit for reliability and security of the network. Note that the parameters in \eqref{eqn:char-discha-Grp2} and \eqref{eqn:emax} for each prosumer needs to be chosen carefully to ensure that prosumers can reap maximum benefits from trading without violating any network limits. These could be decided by a regulatory body for prosumers of particular communities that want to involve in P2P trading. Examples of such regulatory decisions can also be found in other trading mechanisms. For example, in Queensland, Australia, if any prosumer would like to participate in FiT, it cannot install solar panels beyond 5kW capacity according to the regulation of the state. Such regulation has been imposed to enable prosumers to participate in energy trading without compromising network security. Please note that such a regulation should not affect a prosumer'€™s decision of how much energy it would like to trade and how much price it would like to charge per unit of energy. Rather, the purpose of regulation is to offer practical values for constraints like \eqref{eqn:char-discha-Grp2} and \eqref{eqn:emax} that will be beneficial for both prosumers and networks.

The utility functions proposed in  \eqref{Eqn:Utility_Battery_Charge} and \eqref{Eqn:Utility_Battery_DisCharge} possess the property of decreasing marginal benefit, which is ideal for power users~\cite{ Fahrioglu_GT_1999}. In particular, in the decision making process of charging and discharging of the battery, which has a fixed capacity, utility functions with decreasing marginal benefits are capable of modelling the decision making process of prosumers by capturing their benefits, which increases gradually to a certain point of state-of-charge and then decreasing again. Examples of such use of utility functions to model the benefit of power users can also be found in \cite{ChaiBo-TSG:2014, Wayes-J-TSG:2012, Samadi_TSG_Sept_2012}.


Now, based on the utility following \eqref{eqn:Utility_Cost_Type1}, \eqref{Eqn:Utility_Battery_Charge}, and \eqref{Eqn:Utility_Battery_DisCharge} that a prosumer achieves for trading its energy, it forms socially cooperative groups with neighboring prosumers to at different time slots in order to maximize its utility via P2P trading. In the next section, we propose a coalition formation game to capture this decision making process.
\section{Problem Formulation for P2P Trading}\label{sec:GameFormulation}
At any given time slot $t$, the decision of each prosumer $n$ to choose to be either in state 1 or in state 2 may potentially be affected by its production of solar energy, energy demand, energy price, its urgency of using energy for particular tasks, and how much utility it can achieve for interacting with other neighboring prosumers. Note that the decision of being in a particular state is taken by each prosumer individually without any central control. To this end, what follows is a study of the P2P trading price to identify under which condition a prosumer may become interested to charge or discharge its battery to participate in P2P trading with its neighbors. 
\subsection{Decision of a prosumer to charge and discharge the battery}
\subsubsection{Price condition for charging a battery}For any given price $p_{c}(t)$, the utility $U_{n,c}(t)$ to a prosumer for charging its battery is \eqref{Eqn:Utility_Battery_Charge}, where $\frac{\partial^2 U_{n,c}(t)}{{\partial e_{n,c}(t)}^2}=-k\alpha_n<0$. Thus, the preferred strategy $e_{n,c}(t)^*$ of a prosumer to maximize its utility is
\begin{equation}
e_{n,c}(t)^* = \arg \max U_{n,c}(t),
\label{eqn:ValueEnCh_1}
\end{equation}
and therefore,
\begin{equation}
e_{n,c}(t)^* = \left[\frac{1}{k\alpha_n}\left(k(b_{n}-s_n(t))-p_{l} - p_c(t)\right)\right]^{+},
\label{eqn:ValueEnCh_2}
\end{equation}
where $[\cdot]^{+} = \max(\cdot,0)$. Clearly, when a prosumer wants to charge its battery, $e_{n,c}(t)^*\geq 0$. Therefore, the maximum price $p_{n,c}(t)$ that the prosumer $n$ will to pay to charge its battery is 
\begin{equation}
p_{n,c}(t) = k(b_{n}-s_n(t)) - p_l.\label{eqn:ValueEnCh_4}
\end{equation}
Now, based on \eqref{eqn:ValueEnCh_2} and \eqref{eqn:ValueEnCh_4}, a prosumer's behavior towards charging its battery can be influenced by different variation of P2P price $p_c(t)$ as follows.
\begin{itemize}
\item When $p_c(t) = p_{n,c}(t)$, $e_{n,c}(t)^* = 0$. Therefore, the prosumer $n$ will not charge its battery.
\item When $p_c(t) > p_{n,c}(t)$, $e_{n,c}(t)^* < 0$. Therefore, the prosumer $n$ will not charge its battery.
\item When $p_c(t) < p_{n,c}(t)$, $U_{n,c}(t) >0 $ and $e_{n,c}(t)^* >0$. Therefore, the prosumer $n$ will charge its battery.
\end{itemize}

\subsubsection{Price condition for discharging a battery} Similarly, a prosumer $n$ receives utility $U_{n,d}(t)$ by discharging $e_{n,d}(t)$ from its battery. $U_{n,d}(t)$ attains its maximum value when $\frac{\partial U_{n,d}(t)}{{\partial e_{n,d}(t)}}=0$. Therefore,
\begin{equation}
e_{n,d}(t)^* = \arg\max U_{n,d}(t) = \left[\frac{1}{k\alpha_n}\left(ks_n(t)-p_l+p_d(t)\right)\right]^{+}.
\label{eqn:ValueEnDis_1}
\end{equation}
Now, following the same procedure as for the charging price, we can determine the minimum price that can motivate a prosumer to discharge its battery is
\begin{equation}
p_{n,d}(t) = p_l-ks_n(t).
\label{eqn:ValueEnDis_3}
\end{equation}
 To this end, based on \eqref{eqn:ValueEnDis_1} and \eqref{eqn:ValueEnDis_3}, the discharge of a prosumer's battery is influenced by the following variation of P2P price:
\begin{itemize}
\item When $p_d(t)=p_{n,d}(t)$, $e_{n,d}(t)^* = 0$. Therefore, the prosumer $n$ will not be interested to discharge its battery.
\item When $p_ d(t)>p_{n,d}(t)$, $e_{n,d}(t)^* > 0$. Therefore, the prosumer $n$ will discharge its battery.
\item When $p_d(t)<p_{n,d}(t)$, $e_{n,d}(t)^* <0$. Therefore, the prosumer $n$ will not discharge its battery. 
\end{itemize}

Thus, the overall decision of a prosumer to use its battery or not for P2P trading is affected by the value of  P2P charging and discharging prices. A summary of the possible variations and the subsequent decisions are summarized in Table~\ref{table:Table_List_Strategy}. Clearly, a prosumer will choose to be in state $1$ for case 3, whereas for other cases the choice would be state $2$. Note that the decision of energy trading between the batteries of different prosumers is also influenced by respective capacities and SoC of the batteries via \eqref{eqn:ValueEnCh_2} and \eqref{eqn:ValueEnDis_1}. 
\begin{table*}[t]
\centering
\caption{This table lists the strategies of a prosumer on its choice of charging and discharging its battery for the P2P trading purpose and its subsequent choice of state.}
\small
\begin{tabular}{|m{0.5cm}|m{4.5cm}|m{4.5cm}|m{8cm}|}
\hline
\textbf{Case} & \textbf{Value of P2P charging price} $p_c(t)$& \textbf{Value of P2P discharging price} $p_d(t)$ & \textbf{Strategy of prosumer} $n$ \textbf{and Corresponding State}\\
\hline
1 & $p_c(t)<p_{n,c}(t)$ & $p_d(t)\leq p_{n,d}(t)$ & Prosumer $n$ will charge (state $2$).\\
\hline
2 & $p_c(t)\geq p_{n,c}(t)$ & $p_d(t) > p_{n,d}(t)$ & Prosumer $n$ will discharge (state $2$).\\
\hline
3 & $p_c(t)\geq p_{n,c}(t)$ & $p_d(t)\leq p_{n,d}(t)$  & Prosumer $n$ will neither charge not discharge (state $1$).\\
\hline
4 & $p_c(t)<p_{n,c}(t)$ & $p_d(t)>p_{n,d}(t)$  & Prosumer $n$ will either charge or discharge based on it utility according to \eqref{Eqn:Utility_Battery_Charge} and \eqref{Eqn:Utility_Battery_DisCharge} (state $2$).\\
\hline
\end{tabular}
\label{table:Table_List_Strategy}
\end{table*}
\subsection{Coalition formation framework}
A coalition formation game $\Gamma$ can be formally defined by a pair $\Gamma = (\mathcal{N},\nu)$, in which $\mathcal{N}$ is the set of all participating players (i.e., prosumers in this case) and the value of coalition $\nu$. $\nu$ assigns a real number to every coalition $\mathcal{S}\subset\mathcal{N}$ for participating in P2P trading. Under the coalition formation game framework, participating prosumers form coalitions with one another to improve their respective utilities. In doing so, a prosumer may decide to leave a coalition to join a new coalition if its utility is improved by joining the selected new coalition. Indeed, a prosumer may also choose to take its own strategy without being a part of any coalition if that is more beneficial. Now, in order to study the decision making process of each prosumer to choose a certain coalition for P2P trading purpose, we first define the term \emph{Pareto order}. 
\begin{definition}
Consider two collections of coalition $\mathcal{A}$ and $\mathcal{B}$ of the same set of prosumers $\mathcal{N}$. For a prosumer $n\in\mathcal{N}$, $u_\mathcal{A}(n)$ and $u_\mathcal{B}(n)$ denote the utility to prosumer $n$ within coalition $\mathcal{A}$ and $\mathcal{B}$ respectively. Now, coalition $\mathcal{A}$ is preferred by prosumer $n$,~$\forall n\in\mathcal{N}$, over the coalition $\mathcal{B}$ by Pareto order, indicated by $\mathcal{A}\triangleright\mathcal{B}$, if $u_\mathcal{A}(n)\geq u_\mathcal{B}(n)$, with an equality for at least one player.
\label{def:definition_1}
\end{definition} 
Thus, Pareto order bases the preference on the individual payoffs to the prosumers rather than the coalition value $\nu$. Now, based on the Pareto order, each prosumer $n$ decides on the coalition it wants to form in each time slot during a day with the purpose to maximize its utility. To do so, we assume that each prosumer uses \emph{merge and split} rules for forming coalitions with one another.
\begin{figure}[t]
\centering
\includegraphics[width=0.6\columnwidth]{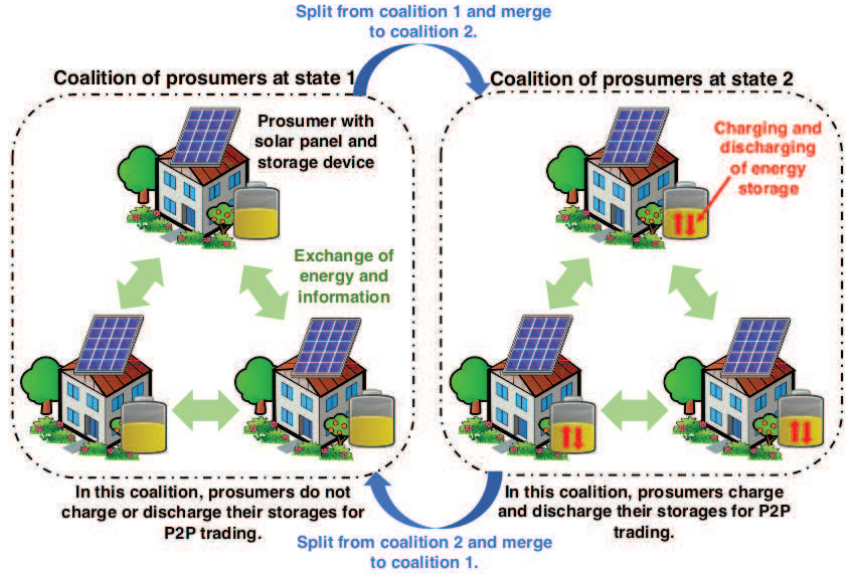}
\caption{This figure shows the social cooperation of prosumers of similar states (within a coalition) and their transfer between different coalitions as outlined in this study.}
\label{fig:coalitionformation}
\end{figure}
\subsubsection*{Merge and split rules} Merge and split rules are two fundamental rules for forming or breaking coalitions based on Pareto order~\cite{Krzysztof_GTR_2009}. They can be defined as follows.
\begin{definition}
A collection of disjoint coalition $\{\mathcal{A}_1, \mathcal{A}_2, \hdots, \mathcal{A}_i\},~\mathcal{A}_i\subset\mathcal{S}$ can agree to merge into a single coalition $\mathcal{B} =\cup_{n=1}^i \mathcal{A}_n$, if the new coalition $\mathcal{B}$ is preferred over $\{\mathcal{A}_1, \mathcal{A}_2, \hdots, \mathcal{A}_i\}$, i.e., 	$\mathcal{B}\triangleright\{\mathcal{A}_1, \mathcal{A}_2, \hdots, \mathcal{A}_i\}$, by the players according to the Pareto order described in Definition~\ref{def:definition_1}.
\label{def:definition_2}
\end{definition}
\begin{definition}
A coalition $\mathcal{B} =\cup_{n=1}^i \mathcal{A}_n,~\mathcal{A}_n\subset\mathcal{S}$ can split into smaller coalitions $\{\mathcal{A}_1, \mathcal{A}_2, \hdots, \mathcal{A}_i\}$ if the resulting coalitions are preferred by the players over $\mathcal{B}$, i.e., $\{\mathcal{A}_1, \mathcal{A}_2, \hdots, \mathcal{A}_i\}\triangleright\mathcal{B}$ , as per the Pareto order in Definition~\ref{def:definition_1}.
\label{def:definition_3}
\end{definition}
The reorganisation of prosumers in different coalitions following the merge-and-split rule is usually conducted in multiple iterations. In each iteration, all the participating prosumers engage to form new coalitions to maximize their utilities. A graphical representation of how prosumers choose different coalition based on merge and split rule in the proposed system is shown in Fig.~\ref{fig:coalitionformation}. 
\subsection{Coalition formation algorithm}
\begin{figure}[t]
\centering
\includegraphics[width=0.6\columnwidth]{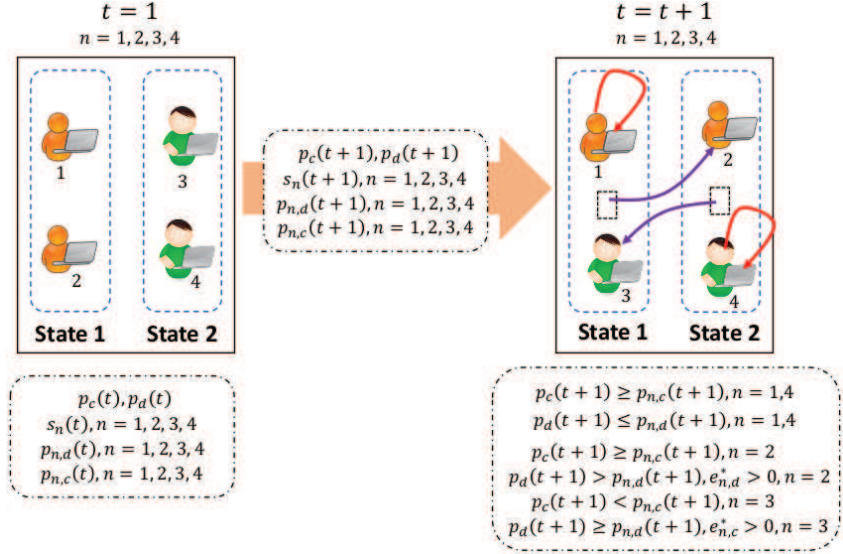}
\caption{This figure shows an example of how four prosumers decide on their respective states based on the price and available energy when they they move from one time slot to the next.}
\label{fig:Algorithm}
\end{figure}
\begin{algorithm}[t!]
\caption{Algorithm to form a stable coalition structure for P2P trading.}
\label{alg:algorithm1}
\begin{algorithmic}[1]
\footnotesize
\FOR{time slot $t=1$ to $T$}
\STATE Set P2P discharge price $p_d(t) = \frac{p_b(t) + p_s(t)}{2}$.
\STATE Set P2P charge price $p_c(t) = (1+\beta)p_d(t)$.
\FOR{Each prosumer $n\in\mathcal{N}$}
\STATE Determine the solar generation $E_{n,\text{pv}}(t)$. 
\STATE Determine energy demand $E_{n,\text{dem}(t)}$ and battery SoC $s_n(t)$.
\IF{$E_{n,\text{pv}}(t)>E_{n,\text{dem}}(t)$}
\STATE $E_{n,\text{cha}}(t) = {\min}((b_n-s_n(t)), e_\text{max}, (E_{n,\text{pv}}(t)-E_{n,\text{dem}}(t)))$.
\STATE $s_n(t) = s_n(t-1)+\eta E_{n,\text{cha}}(t) $.
\STATE $E_{n,\text{sur}}(t) = E_{n,\text{pv}}(t)-(E_{n,\text{dem}}(t)+E_{n,\text{cha}}(t))$.
\ELSIF{$E_{n,\text{pv}}(t)<E_{n,\text{dem}}(t)$}
\STATE $E_{n,\text{dis}}(t) = {\min}((s_n(t)-s_{n,\text{min}}), e_\text{max}, (E_{n,\text{dem}}(t)-E_{n,\text{pv}}(t)))$.
\STATE $s_n(t) = s_n(t-1)-\eta E_{n,\text{dis}}(t) $.
\STATE $E_{n,\text{def}}(t) = E_{n,\text{dem}}(t)-(E_{n,\text{pv}}(t)+E_{n,\text{dis}}(t))$.
\ENDIF
\IF{$E_{n,\text{def}}(t)>0$}
\STATE Prosumer $n$ chooses to be type $1$ of set $\mathcal{W}$.
\ELSE
\STATE Determine $p_{n,c}(t)$ and $p_{n,d}(t)$ following \eqref{eqn:ValueEnCh_4} and \eqref{eqn:ValueEnDis_3} respectively.
\STATE Each prosumer follow Table~\ref{table:Table_List_Strategy} to \emph{initially} decide whether it wants to be in $\mathcal{W}$ or $\mathcal{V}$ prosumer.
\STATE For type $2$, providers and receivers become elements of set $\mathcal{K}$ and $\mathcal{L}$ respectively, i.e., $\mathcal{V} = \mathcal{K}\cup\mathcal{L}$.
\ENDIF
\ENDFOR
\STATE Total available energy from $\mathcal{K}$ is $\sum_{n=1}^K e_{n,d}(t)$.
\STATE Total demand from $\mathcal{L}$ is $\sum_{m=1}^L e_{m,c}(t)$.
\IF{$\sum_{n=1}^K e_{n,d}(t)>\sum_{m=1}^L e_{m,c}(t)$}
\STATE All prosumers in $\mathcal{L}$ finally chooses type 2.
\STATE Calculate excess energy $E_{\text{ex},s}(t) = \sum_{n=1}^K e_{n,d}(t)-\sum_{m=1}^L e_{m,c}(t)$.
\STATE Set $e_{n,d}(t) = e_{n,d}(t)-\min(e_{n,d}(t),\frac{E_{\text{ex},s}(t)}{K})$,~$\forall n\in\mathcal{K}$.
\IF {$e_{n,d}(t)>0$}
\STATE Prosumer $n$ decides to be $\mathcal{K}$ in type 2.
\ELSE\STATE Prosumer $n$ decides to be in $\mathcal{W}$.
\ENDIF
\ELSIF{$\sum_{n=1}^K e_{n,d}(t)<\sum_{m=1}^L e_{m,c}(t)$}
\STATE All prosumers in $\mathcal{K}$ finally chooses type 2.
\STATE Calculate excess demand $E_{\text{ex},d}(t) = \sum_{m=1}^L e_{m,c}(t)-\sum_{n=1}^K e_{n,d}(t)$.
\STATE Set $e_{m,c}(t) = e_{m,c}(t)-\min(e_{m,c}(t),\frac{E_{\text{ex},d}(t)}{L})$,~$\forall m\in\mathcal{L}$.
\IF {$e_{m,c}(t)>0$}
\STATE Prosumer $n$ decides to be $\mathcal{L}\subset\mathcal{V}$.
\ELSE\STATE Prosumer $n$ decides to be in $\mathcal{W}$.
\ENDIF
\ELSE\STATE All prosumers in $\mathcal{L}$ and $\mathcal{K}$ belong to type 2.
\ENDIF\\
\STATE Each prosumer $n$ that was in $\mathcal{N}/\mathcal{V}$ in the previous time slot and decides to be in $\mathcal{V}$,  splits from $\mathcal{N}/\mathcal{V}$ and merges to the prosumers in $\mathcal{V}$.
\STATE Each prosumer $n$ that was in $\mathcal{N}/\mathcal{W}$ in the previous time slot and decides to be in $\mathcal{W}$,  splits from $\mathcal{N}/\mathcal{W}$ and merges to the prosumers in $\mathcal{W}$.
\STATE A stable coalition structure for P2P trading is achieved for time slot $t$.
\ENDFOR
\end{algorithmic}
\end{algorithm}
The formation of coalitions between different prosumers for P2P trading is influenced by a prosumer's decision to socially interact with other prosumers for P2P energy trading. To take the decision, in time slot $t$, a prosumer first meets its demand by energy from its solar panels and storage. Second, it calculates its available surplus, deficiency, and the SoC of the battery. Third, based on the price threshold of the prosumer from \eqref{eqn:ValueEnCh_4} and \eqref{eqn:ValueEnDis_3} and the available energy prices for P2P trading, which is determined by a mid-market rate rule \cite{Tushar_Access_Oct_2018}, a prosumer decides whether it wants to charge or discharge its battery to participate in P2P trading. In other words, each prosumer $n$ decides whether it wants to be in state $1$ or $2$ and utilize the properties of Pareto order to choose the appropriate coalition.

Prosumers that are in state 1 form coalition 1 (Fig.~\ref{fig:coalitionformation}) and conduct P2P trading among themselves using the mechanism proposed in \cite{Tushar_Access_Oct_2018} via a mid-market pricing scheme. As for prosumers in state $2$, each prosumer needs to decide either be a provider or a receiver of energy and thus be in the set $\mathcal{K}$ or $\mathcal{L}$ respectively. Once the decision is taken by prosumers, as shown in Fig.~\ref{fig:Algorithm}, each prosumer $n$ who was in state $1$ in the previous time slot $t-1$ and would like to be in state $2$ at $t$, splits itself from its previous coalition and merges with the players who decide to be in state $2$ and thus becomes a part of a new coalition. A similar transition can also happen for state $2$ players. Finally, P2P trading between the receivers and providers within coalition $2$ is conducted using prices derived in \eqref{eqn:case1_pd} and \eqref{eqn:case1_pc}. The detail of the proposed Pareto order based step-by-step coalition formation algorithm is shown in Algorithm \ref{alg:algorithm1}.

As can be seen from Algorithm \ref{alg:algorithm1}, the formation of a coalition is dynamic in nature. That is, the coalition formation algorithm runs in every time slot to reconfigure the members of each coalition that would achieve the maximum benefit for the prosumers. The correlation between different time slots is captured through the state of charge of storage devices that varies across different time slots according to its charging and discharging pattern. Of course, it is also possible that the same coalition structure may hold for multiple consecutive time slots if Algorithm \ref{alg:algorithm1} identifies that structure is beneficial for the relevant time slots.

Here, it is important to note that, in the proposed coalition formation structure, there is no coalition of prosumers who do not charge/discharge batteries with who will discharge/charge. This is due to the fact that, as the proposed scheme is designed, the charging and discharging price needs to satisfy the conditions explained in Table~\ref{table:Table_List_Strategy}, where the threshold price $p_{n,c}(t)$ and $p_{n,d}(t)$ include the associated charging and discharging degradation cost of batteries. The motivation behind the choice of such prices is the fact that when a prosumer may want its battery to be charged or discharged for trading in the peer-to-peer market, the trading price should reflect the associated battery degradation cost due to charging and discharging, and at the same time, have enough incentive to make the price beneficial for prosumers to put their batteries into the market. Now, prosumers who participate in peer-to-peer trading market without battery do not need to consider such degradation cost and therefore the trading price does not include the additional item to compensate for the cost of battery degradation \cite{Long_Conf_2017, Tushar_Access_Oct_2018}. Given this context, as the proposed model is designed, these prices will always fall within the case 3 of Table \ref{table:Table_List_Strategy}, which subsequently removes the possibility of a trade between prosumers who do not charge/discharge batteries with prosumers who will discharge/charge at the same time.

Indeed, other pricing schemes may regulate the trading in a different way, where the pricing conditions may allow prosumers who do not charge/discharge batteries with prosumers who will discharge/charge at the same time. Meanwhile, the proposed scheme is still valid in terms of proving the importance of social cooperation between prosumers to achieve a prosumer-centric solution, while still participate in peer-to-peer trading. Further, the proposed scheme also help prosumers to decide when they should use batteries to participate in the trading and when they should participate without putting their batteries into the market. 

\subsection{Trading of energy} Once a set of stable coalitions are formed following Algorithm~\ref{alg:algorithm1} in each round of the game, all prosumers in each coalition begins to trade their energy with one another. We have discussed how the prosumers in a coalition of state $1$ players do P2P trading in \cite{Tushar_Access_Oct_2018}. As for P2P energy trading prosumers at state $2$, the energy trading parameters can be determined as follows.

The total energy supplied by $K$ providers for the P2P trading is
\begin{equation}
E_d(t) =\sum_{n=1}^K e_{n,d}(t)^*= \sum_{n=1}^K \frac{1}{k\alpha_{n,d}}\left(ks_{n,d}(t) - p_l + p_d(t)\right).
\label{eqn:analysis_1}
\end{equation}
Without loss of generality, it can be considered that $\alpha_{n,d} = \alpha_d,~\forall n\in\mathcal{K}$. Hence, \eqref{eqn:analysis_1} can be expressed in its simple form as
\begin{equation}
E_d(t) = \frac{\sum_{n=1}^K s_{n,d}(t)}{\alpha_{d}} - \frac{(p_l-p_d(t))K}{k\alpha_{d}}.
\label{eqn:analysis_2}
\end{equation}
Similarly, total amount of energy received by the receivers of the coalition can be expressed by
\begin{equation}
E_c(t) =\sum_{m=1}^L e_{m,c}(t)^*= \frac{\sum_{m = 1}^L (b_{m,c}(t)-s_{m,c}(t))}{\alpha_c}-\frac{(p_l + p_c(t))L}{k\alpha_c},
\label{eqn:analysis_3}
\end{equation}
assuming $\alpha_{m,c} = \alpha_c.~\forall m\in\mathcal{L}$. 

Now, let us consider that the P2P charging price $p_c(t)$ and discharging price $p_d(t)$ per unit of energy are related to one another via~\cite{RongYu_ITJ:2014}
\begin{equation}
p_c(t) = (1+\beta)p_d(t),
\label{eqn:price_analysis}
\end{equation}
where $\beta>0$. The relationship in \eqref{eqn:price_analysis} is motivated by the fact that a prosumer needs to pay the government tax and a subscription fee, captured via $\beta$, for using the transmission line provide by the grid for P2P purpose. Now, the flow of energy for charging and discharging batteries via P2P establishes that the total energy charged by the receivers batteries should be equal to the total energy discharged by the providers in $\mathcal{V}$, i.e., $E_d(t) = E_c(t)$. 

To this end, by equating \eqref{eqn:analysis_2} to \eqref{eqn:analysis_3} and replacing the value of $p_c(t)$ with $p_d(t)$ from \eqref{eqn:price_analysis}, the trading prices $p_c(t)$ and $p_d(t)$ for prosumers at state $2$ can be calculated as:
\begin{align}
\small
p_{d,s2}(t) = \frac{p_l(\alpha_c K - \alpha_d L) - k(\alpha_c\displaystyle\sum_{n=1}^{K}s_{n,d}(t) - \alpha_d\displaystyle\sum_{m=1}^{L}(b_{m,c}(t) - s_{m,c}(t)))}{\alpha_c K +\alpha_d(1+\beta)L};
\label{eqn:case1_pd}
\end{align}
\begin{align}
p_{c,s2}(t)=\frac{p_l(\alpha_c K - \alpha_d L) - k(\alpha_c\displaystyle\sum_{n=1}^{K}s_{n,d}(t) - \alpha_d\displaystyle\sum_{m=1}^{L}(b_{m,c}(t) - s_{m,c}(t)))}{\frac{\alpha_c K +\alpha_d(1+\beta)L}{(1+\beta)}}.
\label{eqn:case1_pc}
\end{align}
Thus, $p_{d,s2}(t)$ and $p_{c,s2}(t)$ are the discharging and charging prices used by prosumers at state 2 to trade their battery energy over the P2P network.
\begin{remark}
Please note that in the manuscript we do not consider the scenario, in which a prosumer at state 2 can sell its battery energy to the CPS. This mainly due to the fact that such a trading model relies on a pricing scheme set by the CPS that provides prosumers with very limited benefits \cite{Tushar_SPM_2018}. Consequently, as articulated in \cite{Weule_ABCNews_2018}, battery at residential premises is not economically viable. Nonetheless, if a CPS can participate in the local market as a peer and give prosumers to negotiate the price for per unit energy it would sell to the CPS, the proposed model can easily be extended by including CPS as a peer. In such cases, a prosumer can also sell its energy to the CPS through P2P mechanism.
\end{remark}
\section{Properties of the Coalition Formation}\label{sec:properties}
Given the proposed coalition formation algorithm and resultant coalition structure, we now study the properties of the proposed scheme.
\subsection{Stability and optimality}In this section, we investigate the stability and optimality of the proposed coalition formation. To this end, first we define the stability and optimality from \cite{Saad-coalition:2009} and \cite{RongYu_ITJ:2014} as follows. 
\begin{definition}
A group of coalitions is said to be stable, if no prosumer has an interest to perform a merge-and-split operation in order to form another new coalition for better payoff in a selected time slot. This is known as $\mathbb{D}_{hp}$ stable.
\label{def:definition_4}
\end{definition}

%
\begin{definition}
A partition refers to the Pareto optimal network structure if it exhibits the property of $\mathbb{D}_c$ stability with the following characteristics:
\begin{enumerate}
\item The partition is $\mathbb{D}_{hp}$ stable.
\item The resultant partition is the unique outcome of any round of merge-and-split operation.
\item The partition maximizes the social welfare (sum utilities of the participating prosumers).
\end{enumerate}
\label{def:definition_5}
\end{definition}
Second, we note that in every time slot, as the $\Gamma$ is designed, a prosumer $n$ always decides whether it wants to charge or discharge its battery for P2P trading by following scenarios listed in Table~\ref{table:Table_List_Strategy}. This could lead to one of the following two outcomes: 1) Prosumer $n$ may decides either to be in state $2$ and puts its battery into the market for P2P trading, or 2) Prosumer $n$ decides to be in state $1$ and participates in P2P trading without putting its battery into the market. Thus, for case $1$, prosumer $n$ belongs to the coalition of prosumers that are in state $2$. As for case $2$, a prosumer may decide either to form a coalition with other state $1$ prosumers or it may participate in the energy market non-cooperatively (trade with the CPS alone). Nonetheless, it is shown in \cite{Lee_JSAC_2014} that for the considered utilities in \eqref{eqn:Utility_Cost_Type1}, it is always beneficial for prosumers at state $1$ to form a stable group with one another, rather than participate non-cooperatively. Therefore, it is reasonable to state the following Proposition~\ref{Proposition:1}.
\begin{proposition}
In the proposed $\Gamma$, at any time slot $t$, a prosumer $n\in\mathcal{N}$ always forms a coalition with prosumers that are either in state $1$ or state $2$  following the Pareto order defined in Definition~\ref{def:definition_1}, and never trades non-cooperatively.
\label{Proposition:1}
\end{proposition}
\begin{theorem}
At any given time slot $t$, the network structure or partitions resulting from the proposed $\Gamma$ is stable and Pareto optimal.
\label{theorem:theorem2}
\end{theorem}
\begin{proof}
According to Algorithm~\ref{alg:algorithm1}, we note that, at any time slot $t$, prosumers in $\mathcal{L}$ and $\mathcal{K}$ charge and discharge their batteries for P2P trading respectively, where
\begin{equation}
\mathcal{L} = \left\{n:e_{n,c}^{*}(t)> 0~\text{and}~p_c(t)<p_{n,c}(t), p_{d}(t)\geq p_{n,d}(t)\right\}\label{eqn-a1}
\end{equation} and
\begin{equation}
\mathcal{K} = \left\{n:e_{n,d}^{*}(t)>0~\text{and}~p_c(t)\geq p_{n,c}(t), p_{d}(t)> p_{n,d}(t)\right\}.\label{eqn-a2}
\end{equation}
When $p_c(t)<p_{n,c}(t)$ and $p_d(t)>p_{n,d}(t)$, the decision of a prosumer $n\in\mathcal{N}$ to be in $\mathcal{L}$ or $\mathcal{K}$ is determined by the maximum utillity that it achieves from charging and discharging. That is
\begin{eqnarray}
n\in\begin{cases}
\mathcal{L} & \text{if}~ \left . U_{n,c}(t)\right\vert_{e_{n,c}^*(t)\geq 0} > \left . U_{n,d}(t)\right\vert_{e_{n,d}^*(t)\geq 0}\\
\mathcal{K} & \text{if}~ \left . U_{n,d}(t)\right\vert_{e_{n,d}^*(t)\geq 0} > \left . U_{n,c}(t)\right\vert_{e_{n,c}^*(t)\geq 0}
\end{cases}.\label{eqn-a3}
\end{eqnarray}
The rest of the prosumers $\mathcal{N}\setminus\left(\mathcal{L}\cup\mathcal{K}\right)$ remain in state 1 and belong to $\mathcal{W}$. According to \eqref{eqn-a1}, \eqref{eqn-a2}, and \eqref{eqn-a3}, indeed, the choice of  $\mathcal{K}, \mathcal{L}$ (coalition of state 2) and $\mathcal{W}$ (coalition of state 1) by a prosumer $n$ is determined based on the achieved maximum utility by the prosumers (that is, Pareto order in Definition~\ref{def:definition_1}), and hence $n$ cannot be better paid off by choosing an alternative coalition. Consequently, it will have no incentive to split from its current coalition and merge to a new coalition for a better payoff. Therefore,  following Definition~\ref{def:definition_4}, the resultant network structure is $\mathbb{D}_{hp}$ stable.

Further, as $\Gamma$ is designed, at any given time slot $t$, the $\mathbb{D}_{hp}$ stable network structure always consists of two coalitions with set of players $\mathcal{W}$ and $\mathcal{V}$, and satisfies $|\mathcal{W}\cup\mathcal{V}| = N$. This unique set of coalitions maximizes the individual benefit to each prosumers, and consequently their sum benefits is also maximized. Hence, the proposed coalition is $\mathbb{D}_c$ stable. 

Now, since the coalition structure is $\mathbb{D}_c$ stable, according to Definition~\ref{def:definition_4}, the formation of coalition is also Pareto optimal. Thus, Theorem~\ref{theorem:theorem2} is proven.
\end{proof}
\subsection{Prosumer-Centric Property}To investigate whether the proposed P2P energy trading scheme satisfies the prosumer-centric property, we exploit a number of models from motivational psychology. 
\begin{definition}
A P2P energy trading scheme is defined to be prosumer-centric if 
\begin{itemize}
\item The coalition structure formed in $\Gamma$ is stable.
\item The utility received by each prosumer from $\Gamma$ satisfies the rational economic, elaboration likelihood and positive reinforcement models of motivational psychology~\cite{MotivationGame_2019}.
\end{itemize}
\end{definition}
Clearly, from Theorem~\ref{theorem:theorem2}, the proposed coalition formation game possesses stability. Hence, to be prosumer-centric, the proposed P2P energy trading need to satisfy the models from motivational psychology.

Essentially, motivation psychology is a branch of behavioral science that studies the impact of human psychological process to initiate real behavior~\cite{Beebe:1999,Miller:2002}. Motivational psychology consists of a number of behavioral models that can be used to understand whether a developed technology can motivate users to accept it. One such application of motivational psychology models in determining the feasibility of attracting users to efficiently use their heating, ventilation, and air conditioning (HVAC) units can be found in \cite{Tushar_MPsy_2018}. 

While different motivational psychology models can be used to validate the impacts of different technologies for attracting users to participate, in this section, we will focus on three particular models, which are relevant to our study, to demonstrate the prosumer-centric property of the proposed trading scheme. To this end, we first introduce the \emph{rational economic}, the \emph{elaboration likelihood} and the \emph{positive reinforcement models}. Then, we study whether the proposed P2P trading scheme satisfy the properties of the motivational psychology models, so as to exhibit the prosumer-centric property.
\subsubsection{Rational economic model}According to \cite{Shipworth:2000}, a prosumer's participation in energy trading is predominantly based on his economically rational decision. That is, monetary benefit is a key motivator for people to be responsible and logical about participating in energy trading in a P2P network.
\subsubsection{Elaboration likelihood model}According to this model, there could be two ways to communicate and motivate people to participate in energy trading: the central path and the peripheral path~\cite{Petty_1986}. The central path is suitable when an individual cares about the issue and can easily access the necessary information. However, he may potentially deviate from his supported position if the subject conveys unfavorable thoughts due to the ambiguity of the message. In such a case, the peripheral path is more appropriate. Essentially, the peripheral path tries to associate the advocated position with things the receiver already thinks positively towards, such as monetary and environmental benefit, using an expert appeal.
\subsubsection{Positive reinforcement model}A positive reinforcement refers to the case when a human response to a circumstance is followed by a reinforcing stimulus that increases the potential of having the same response from the human when a similar situation arises~\cite{Hockenbury_2003}. For example, by always receiving a better utility by cooperating with other peers within a P2P energy network is likely to encourage the users cooperate with its peers for energy trading again in the future.
%
\begin{theorem}
The proposed social cooperation based P2P energy trading scheme is prosumer-centric.
\label{theorem:theorem4}
\end{theorem}
\begin{proof}
To prove the theorem, we note the following: 
\begin{itemize}
\item A prosumer $n$ decides to either charge or discharge its batteries based on the utility functions \eqref{Eqn:Utility_Battery_Charge} and \eqref{Eqn:Utility_Battery_DisCharge}. The utilities to $n$ are dominantly influenced by the buying price $p_{n,d}$ and selling price $p_{n,c}$ as demonstrated in Table~\ref{table:Table_List_Strategy}. In this context, clearly, economic benefit plays key role for a prosumers to choose a stable coalition to perform P2P trading according to Algorithm~\ref{alg:algorithm1}. Thus, the proposed scheme satisfies the rational-economic model.
\item At any time slot $t$, the proposed P2P energy trading is Pareto optimal. Therefore, every time a prosumer $n$ wants to trade energy, participating in P2P energy trading by forming coalitions with other prosumers in the network is always beneficial to $n$, rather than acting noncoopeatively. This subsequently proves that the proposed scheme satisfies positive reinforcement model.
\item The improved net benefit of P2P trading could be an effective way to demonstrate the prosumers of the advantage of participation. Thus, the proposed scheme can easily use this peripheral path to help prosumers understand and convince them to participate in P2P trading. 
\end{itemize}
Thus, the proposed social cooperation based P2P trading satisfies all three considered motivational psychology models and thus exhibits the properties of a prosumer-centric scheme.
\end{proof}

Indeed, charging and discharging of battery cause battery degradation, as discussed in \cite{Wenzl_JPS_2005}. It may discourage some prosumers to use the battery extensively and unwilling to share any battery capacity. Therefore, such prosumers are belong to state 1 in the proposed formulation. Nonetheless, despite battery degradation, battery sharing is a popular phenomenon within the prosumers and they are willing to share their batteries with neighbours as discussed in \cite{Tushar_TSG_2_2016} and \cite{ChaoLong_AE_2018}. Further, as identified by a number of renewable energy service providers, prosumers, at current time, do not emphasize on battery degradation, rather they are more keen to become environmental friendly, reduce their current energy cost, and become grid independent during natural disaster using battery integrated solar system, which has increased the uptake of residential battery storage significantly since 2017~\cite{2018-EM}. Thus, proposed model is valid and important for the decision making process of this set of users.
\section{Case Study}\label{sec:casestudy}In this section, we show some results from numerical case studies to demonstrate the properties of the proposed trading scheme. In particular, we demonstrate 1) how each prosumer may choose to form different coalition that leads to a \emph{stable} coalition framework, 2) how the proposed coalition formation algorithm brings benefit to the prosumers for participating in P2P trading compared non-participating prosumers, and 3) satisfy the prosumer-centric property. For the numerical case study, we use the real-data of solar generation and household energy consumption data available from Redback Technologies. Redback is a Queensland based startup in Australia that provides smart energy solutions to prosumers in Queensland, Victoria, and New South Wales. We use $15$ min sample data to validate the proposed framework and the data used for this case study was collected in February 2018. The values of retailer's time-of-use electricity selling price is also collected from Redback and the FiT price is assumed to be $10$ cents/kWh according to the FiT price used in Brisbane, Australia.
\begin{figure}[t!]
\centering
\includegraphics[width=0.6\columnwidth]{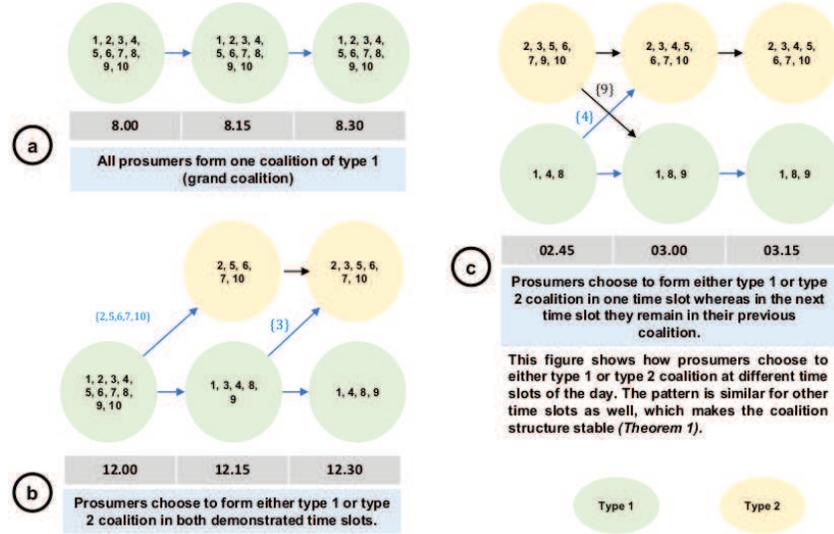}
\caption{This figure demonstrates three different time snaps of a day, in which prosumers within the considered system model decide on various coalition to participate in P2P trading.}
\label{fig:Figure1_Stable}
\end{figure}

\subsubsection{Stability of coalition}In Fig.~\ref{fig:Figure1_Stable}, we show the choice of different coalitions by prosumers at three different group of time slots of a selected day.  Although we choose a selected group of time slots due to the lack of space in the paper to accommodate the demonstration for all $96$ time slots, similar pattern of coalition formation is observed for the rest of the time slots as well. Now, based on this figure, 
\begin{itemize}
\item In Fig.~\ref{fig:Figure1_Stable} (a), all prosumers choose not to charge or discharge their batteries for P2P trading and thus form a grand coalition of state $1$ prosumers in all considered time slots. That is $\mathcal{W} = \{1, 2, \hdots, 10\}$, $\mathcal{V} = \{\phi\}$, and therefore $|\mathcal{W}\cup\mathcal{V}| = |\mathcal{N}| = 10$. 
\item In Fig.~\ref{fig:Figure1_Stable} (b), prosumers form a new set of coalitions in every time slot. However, in all three time slots $|\mathcal{W}\cup\mathcal{V}| = 10$. 
\item Finally, a similar pattern in forming coalition, i.e., $|\mathcal{W}\cup\mathcal{V}| = 10$, is also observed in Fig.~\ref{fig:Figure1_Stable} (c), in which the prosumers form a new coalition during the transition for first time slot to the second whereas they remain in the same coalition in the next time slot. 
\end{itemize}
Thus, in every time slot, the outcome of the game provides a unique outcome, which always contains two coalitions, and once a prosumer chooses a coalition following Algorithm 1, it remains in that coalition until the next time slot of energy trading without any motivation to deviate due to Pareto order (see Definition 1). As such, regardless of whether a new coalition is created in a time slot or not, as the scheme is designed, the coalition structure satisfies Theorem~\ref{theorem:theorem2} and therefore exhibits the property of a stable coalition structure in each time slot. For more details on the stability of coalition formation games, please see \cite{Walid_TWC_Sept_2009}.
\begin{figure}[t]
\centering
\includegraphics[width=0.6\linewidth]{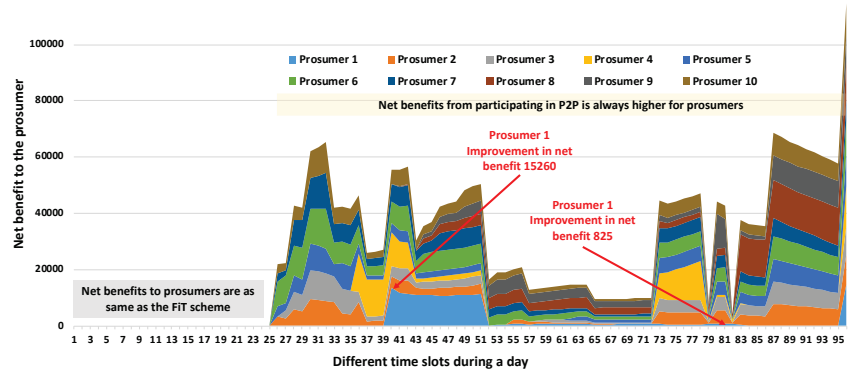}
\caption{Demonstration of net benefits to prosumers for participating in P2P trading via proposed social cooperation framework.}
\label{fig:Fig1}
\end{figure}

\subsubsection{Benefits to prosumers}In Fig.~\ref{fig:Fig1}, we show how the proposed social cooperation framework can improve the net benefits\footnote{Net benefit is the difference between utility and cost of trading energy.} attained by the prosumers. In particular, we show the difference between the net benefits to prosumers with and without social cooperation. By ``without social cooperation", we refer to the FiT scheme, in which a prosumer uses its battery either for storing the excess solar energy from its rooftop solar or for spending the stored energy for household purpose and use only the surplus to trade with the CPS. Now, according to Fig~\ref{fig:Fig1}, first we note that the improvement in net benefits to prosumers for participating in P2P trading vary across both different time slots and across different prosumers. For example, while the improvement in net benefit for prosumer one is as large as $15260$ at time slot $40$, it reduces to as low as $825$ at time slot $82$. In the morning, the improvement in net benefit is zero. The pattern, which is mainly due to the different generation and energy demand pattern of the prosumer at different time slot of the day, is also similarly random for other prosumers as well. The improvement of net benefits for different prosumers are also different for the same reason. Nonetheless, the P2P trading via proposed social cooperation framework is \emph{never detrimental} to any prosumer at any given time slot of the day. As evident from Fig.~\ref{fig:Fig1}, in most of the time slots of the selected day P2P trading demonstrates improvements in net benefits to all the prosumers while the net benefit is at least as good as the FiT scheme to prosumers in rest of the time slots. 

\begin{figure}[t]
\centering
\includegraphics[width=0.6\columnwidth]{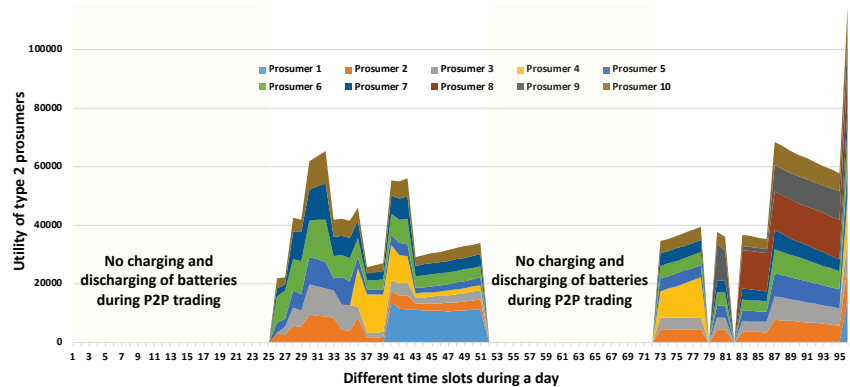}
\caption{Demonstration of opportunistic use of battery by different prosumers while participating in P2P trading.}
\label{fig:Fig2}
\end{figure}
Further, it is also important to note that, as the cooperative framework is designed, the prosumers do not require to always use their batteries for participating in P2P trading. Rather, they can opportunistically choose to trade their battery energy when they find it beneficial for them by following \eqref{Eqn:Utility_Battery_Charge} and \eqref{Eqn:Utility_Battery_DisCharge} and still manage to reap a better (or at least equal) utility compared to the case without cooperation. For instance, as shown in Fig.~\ref{fig:Fig2}, although a number prosumers are interested to trade their battery energy with other peers of the network (at time slots  $25$ to $45$, for example), there is no trade of battery energy from any prosumers during time slots from $50$ to $60$. The benefit for P2P trading is, however, still obvious for time slots $50$ to $60$ as noted in Fig.~\ref{fig:Fig1}. Such opportunistic use of battery helps prosumers to intelligently use their batteries for P2P trading without significantly compromising the lifetime of the batteries. Further, the proposed scheme may help prosumers to estimate how much they may need to invest on their batteries based on their frequency of battery usage for trading purposes (e.g., a prosumer do not need to buy a large battery if the result shows that it would be under used).

\begin{table}[t]
\centering
\caption{This table illustrates how the proposed P2P energy trading scheme exhibits the prosumer-centric property.}
\includegraphics[width=0.8\columnwidth]{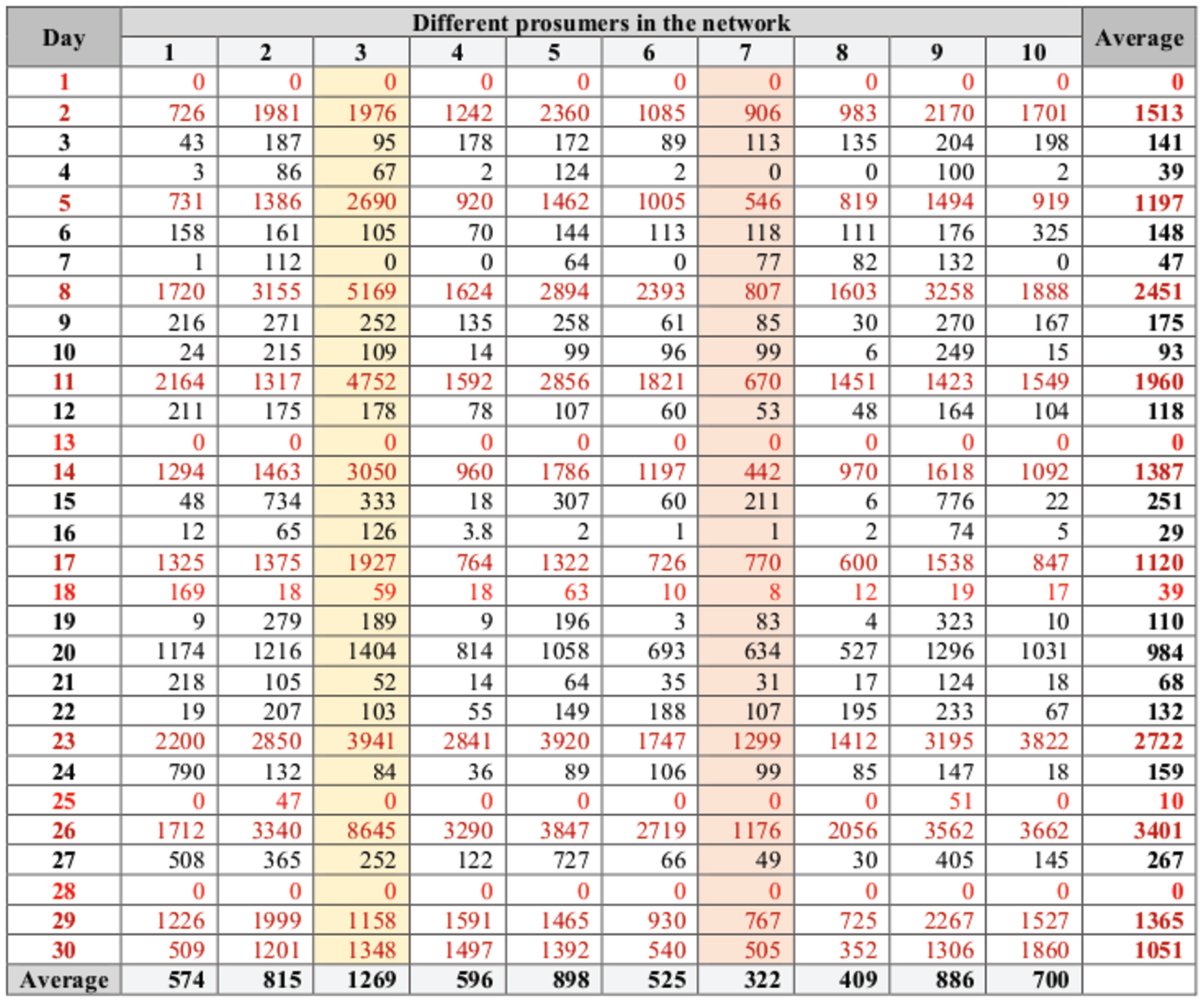}
\label{fig:Fig3}
\end{table}

\subsubsection{Prosumer-centric property}
To show that the proposed P2P trading scheme is prosumer-centric, we show the net benefits of ten prosumers that they achieve by cooperating with one another in P2P trading compared to the case without cooperation in Table~\ref{fig:Fig3}. In the table, we have demonstrated the result for a selected month for the visual clarity of the demonstrated numbers. However, a similar result holds for any other of the year. Now, according to Table~\ref{fig:Fig3}, the benefit to each prosumer for cooperating with one another is more beneficial most of the time and never detrimental compared to participating in the FiT scheme. This characteristic complies with the rational economic and elaboration likelihood properties of the motivational psychology model. Indeed, outcomes of trading can be affected by both weather of the considered day and the type of the prosumer. For example, for prosumer 2, 3, 5, and 9, social cooperation brings higher average benefits compared to other prosumers within the system. On the other hand, the average benefit to prosumer 7 and 8 relatively lower. Similar impacts on the outcome are also observed for different days due to weather conditions. Days, when sunshine hours are relatively longer P2P trading, are proven to be more beneficial. Examples such days in this experiment include days 2, 5, 8, 11, 14, 17, 23, 26, and 29. However, if sunshine hours are limited such as in days 4, 10, 16, and 21, the benefits are lower. Note that for a completely cloudy day, the benefit of P2P is as same as the FiT scheme. This is due to the fact that prosumers do not have any surplus to sell these days. Rather, they use their energy from storage (no solar generation) to meet their own demand only. Nonetheless, the attained net benefit per prosumer shows a consistent performance improvement on sunny days and similar performance on cloudy days when compared to the FiT scheme. Thus, as per definition, it satisfies the positive reinforcement property. Hence, the proposed scheme exhibits prosumer-centric property.

\subsubsection{Computational complexity}Now, we discuss the computational complexity of the proposed scheme. The main computation complexity of the proposed scheme, however, stems from the iterative decision-making process of each prosumer on under which coalition it would like to participate in P2P trading. However, as the proposed scheme is designed, the decision made by each prosumer is determined by a set of simple rules in Table II based on the information already available to the prosumer and through very simple calculations in \eqref{eqn:case1_pd} and \eqref{eqn:case1_pc}. As a result, the computation burden is negligible.

As for the decision of trading parameters, the computation complexity is greatly reduced by the split of overall trading into two different coalitions. In coalition 1, which consists of prosumers at state 1, the computational complexity is very minimal as each prosumer's utility and cost are decided by multiplications of already given parameters. For coalition 2, on the other hand, prices are defined by two closed-form expressions with given parameters for each time slot. Thus, here is the computational burden is also negligible. 

Of course, for a very large number of prosumers, the computation complexity may increase. However, due to the above-mentioned reasons, which would be true across any network, the computational complexity would feasible for adopting the proposed scheme.

\section{Conclusion}\label{sec:conclusion}In this paper, we have studied a peer-to-peer energy trading scheme by exploiting the social cooperation between different prosumers of the network. For this purpose, we have proposed a coalition formation game that can help each participating prosumer to opportunistically decide whether it should put its battery in the peer-to-peer market for energy trading. It has been shown that  the coalition structure that stems from prosumers' social cooperation at each time slot is stable and the resultant peer-to-peer trading scheme is prosumer-centric. Further, we have conducted a number of case studies for the proposed scheme based on Australia based household energy usage and solar generation data and provided some numerical results to show that the proposed scheme can enable prosumers without any storage to participate and still be benefitted from peer-to-peer trading. Further, it has been demonstrated that the proposed trading scheme has the potential to help prosumers to intelligently use their batteries for peer-to-peer trading while consider the degradation cost of their battery. Here it is important to note the economic benefit could be different and more for the feed-in-tariff scheme if the feed-in-tariff rate is significantly large. For example, when feed-in-tariff was first offered to the prosumers in Queensland, Australia, the rate was $44$ cents per kWh. Nevertheless, such a high rate is not being offered anymore and the rate has reduced significantly as articulated in numerous existing studies.

A potential extension of the proposed research is to investigate the impact of such peer-to-peer trading on bus voltages as well as on the overall losses of the network.

\section*{Acknowledgement}This research was funded by the Queensland Government through the Advance Queensland Research Fellowship AQRF11016-17RD2.  

\begin{thebibliography}{10}
\providecommand{\url}[1]{#1}
\csname url@samestyle\endcsname
\providecommand{\newblock}{\relax}
\providecommand{\bibinfo}[2]{#2}
\providecommand{\BIBentrySTDinterwordspacing}{\spaceskip=0pt\relax}
\providecommand{\BIBentryALTinterwordstretchfactor}{4}
\providecommand{\BIBentryALTinterwordspacing}{\spaceskip=\fontdimen2\font plus
\BIBentryALTinterwordstretchfactor\fontdimen3\font minus
  \fontdimen4\font\relax}
\providecommand{\BIBforeignlanguage}[2]{{%
\expandafter\ifx\csname l@#1\endcsname\relax
\typeout{** WARNING: IEEEtran.bst: No hyphenation pattern has been}%
\typeout{** loaded for the language `#1'. Using the pattern for}%
\typeout{** the default language instead.}%
\else
\language=\csname l@#1\endcsname
\fi
#2}}
\providecommand{\BIBdecl}{\relax}
\BIBdecl

\bibitem{Sousa_RSER_Apr_2019}
T.~Sousa, T.~Soares, P.~Pinson, F.~Moret, T.~Baroche, and E.~Sorin,
  ``{Peer-to-peer and community-based markets: A comprehensive review},''
  \emph{Renewable and Sustainable Energy Reviews}, vol. 104, pp. 367--378, Apr.
  2019.

\bibitem{Peck_Spectrum_2017}
M.~E. Peck and D.~Wagman, ``Energy trading for fun and profit buy your
  neighbor's rooftop solar power or sell your own-it'll all be on a
  blockchain,'' \emph{IEEE Spectrum}, vol.~54, no.~10, pp. 56--61, Oct. 2017.

\bibitem{Liang-ChengYe-AE:2017}
L.-C. Ye, J.~F.~D. Rodrigues, and H.~X. Lin, ``Analysis of feed-in tariff
  policies for solar photovoltaic in china 2011-2016,'' \emph{Applied Energy},
  vol. 203, pp. 496--505, Oct. 2017.

\bibitem{Tushar_SPM_2018}
W.~Tushar, C.~Yuen, H.~Mohsenian-Rad, T.~Saha, H.~V. Poor, and K.~L. Wood,
  ``{Transforming energy networks via peer-to-peer energy trading - The
  potential of game-theoretic approaches},'' \emph{IEEE Signal Processing
  Magazine}, vol.~35, no.~4, pp. 2--24, Jul 2018.

\bibitem{FiTPolicy_2015}
{Strategic Futures, Energy Industry Policy}, ``{Queensland solar bonus scheme
  policy guide},'' Department of Energy and Water Supply, State of Queensland,
  QLD, Australia, Report, 2015.

\bibitem{Jogunola_Energies_Dec_2017}
O.~Jogunola, A.~Ikpehai, K.~Anoh, B.~Adebisi, M.~Hammoudeh, S.-Y. Son, and
  G.~Harris, ``State-of-the-art and prospects for peer-to-peer
  transaction-based energy system,'' \emph{MDPI Energies}, vol.~10, no.~12, pp.
  2106:1--2106:28, Dec. 2017.

\bibitem{Morstyn_NE_2018}
T.~Morstyn, N.~Farrell, S.~J. Darby, and M.~D. McCulloch, ``Using peer-to-peer
  energy-trading platforms to incentivize prosumers to form federated power
  plants,'' \emph{Nature Energy}, vol.~3, pp. 94--101, Feb. 2018.

\bibitem{Si_AE_Dec_2018}
F.~Si, J.~Wang, Y.~Han, Q.~Zhao, P.~Han, and Y.~Li, ``{Cost-efficient
  multi-energy management with flexible complementarity strategy for energy
  internet},'' \emph{Applied Energy}, vol. 231, pp. 803--815, Dec. 2018.

\bibitem{Nguyen_AE_Oct_2018}
S.~Nguyen, W.~Peng, P.~Sokolowski, D.~Alahakoon, and X.~Yu, ``{Optimizing
  rooftop photovoltaic distributed generation with battery storage for
  peer-to-peer energy trading},'' \emph{Applied Energy}, vol. 228, pp.
  2567--2580, Oct. 2018.

\bibitem{ChaoLong_AE_2018}
C.~Long, J.~Wu, Y.~Zhou, and N.~Jenkins, ``Peer-to-peer energy sharing through
  a two-stage aggregated battery control in a community microgrid,''
  \emph{Applied Energy}, vol. 226, pp. 261--276, Sep. 2018.

\bibitem{Melendez_AE_Aug_2019}
K.~A. Melendez, V.~Subramanian, T.~K. Das, and C.~Kwon, ``{Empowering end-use
  consumers of electricity to aggregate for demand-side participation},''
  \emph{Applied Energy}, vol. 248, pp. 372--382, Aug. 2019.

\bibitem{Paark_Sustainability_Mar_2018}
L.~W. Park, S.~Lee, and H.~Chang, ``A sustainable home energy prosumer-chain
  methodology with energy tags over the blockchain,'' \emph{MDPI
  Sustainability}, vol.~10, no.~3, pp. 658:1--658:18, Mar. 2018.

\bibitem{Al-Baz_AE_May_2019}
W.~El-Baz, P.~Tzscheutschler, and U.~Wagner, ``{Integration of energy markets
  in microgrids: A double-sided auction with device-oriented bidding
  strategies},'' \emph{Applied Energy}, vol. 241, pp. 625--639, May 2019.

\bibitem{Moret_TPWRS_EA_2018}
F.~{Moret} and P.~{Pinson}, ``Energy collectives: {A} community and fairness
  based approach to future electricity markets,'' \emph{IEEE Transactions on
  Power Systems}, vol.~34, no.~5, pp. 3994--4004, Sep. 2019.

\bibitem{Sorin_TPWRS_Mar_2019}
E.~{Sorin}, L.~{Bobo}, and P.~{Pinson}, ``Consensus-based approach to
  peer-to-peer electricity markets with product differentiation,'' \emph{IEEE
  Transactions on Power Systems}, vol.~34, no.~2, pp. 994--1004, Mar. 2019.

\bibitem{Moret_PSCC_June_2018}
F.~{Moret}, T.~{Baroche}, E.~{Sorin}, and P.~{Pinson}, ``Negotiation algorithms
  for peer-to-peer electricity markets: Computational properties,'' in
  \emph{Power Systems Computation Conference (PSCC)}, Dublin, Ireland, June
  2018, pp. 1--7.

\bibitem{Cadre_Elsevier_Press_2019}
\BIBentryALTinterwordspacing
H.~L. Cadre, P.~Jacquot, C.~Wan, and C.~Alasseur, ``Peer-to-peer electricity
  market analysis: {F}rom variational to generalized nash equilibrium,''
  \emph{European Journal of Operational Research}, 2019, in press. [Online].
  Available: \url{https://doi.org/10.1016/j.ejor.2019.09.035}
\BIBentrySTDinterwordspacing

\bibitem{Siozios_DATE_Mar_2019}
K.~{Siozios} and S.~{Siskos}, ``A low-complexity framework for distributed
  energy market targeting smart-grid,'' in \emph{Proc. of Design, Automation
  Test in Europe Conference Exhibition (DATE)}, Florence, Italy, Mar. 2019, pp.
  878--883.

\bibitem{Chapman_TSG_2018}
\BIBentryALTinterwordspacing
J.~{Guerrero}, A.~C. {Chapman}, and G.~{Verbi\v{c}}, ``Decentralized {P2P}
  energy trading under network constraints in a low-voltage network,''
  \emph{IEEE Transactions on Smart Grid}, 2018. [Online]. Available:
  \url{10.1109/TSG.2018.2878445}
\BIBentrySTDinterwordspacing

\bibitem{Guerrero_AUPEC_Nov_2017}
J.~{Guerrero}, A.~{Chapman}, and G.~{Verbic}, ``A study of energy trading in a
  low-voltage network: Centralised and distributed approaches,'' in \emph{Proc.
  of Australasian Universities Power Engineering Conference (AUPEC)},
  Melbourne, Australia, Nov. 2017, pp. 1--6.

\bibitem{Hamada_AE_Apr_2019}
H.~Almasalma, S.~Claeys, and G.~Deconinck, ``{Peer-to-peer-based integrated
  grid voltage support function for smart photovoltaic inverters},''
  \emph{Applied Energy}, vol. 239, pp. 1037--1048, Apr. 2019.

\bibitem{Zhang_AE_Feb_2019}
X.~Zhang, S.~Zhu, J.~He, and B.~Y.~X. Guan, ``Credit rating based real-time
  energy trading in microgrids,'' \emph{Applied Energy}, vol. 226, pp.
  985--996, Feb. 2019.

\bibitem{Werth_TSG_2018}
A.~Werth, A.~Andre, D.~Kawamoto, T.~Morita, S.~Tajima, D.~Yanagidaira,
  M.~Tokoro, and K.~Tanaka, ``Peer-to-peer control system for dc microgrids,''
  \emph{IEEE Transactions on Smart Grid}, vol.~PP, no.~99, pp. 1--1, 2016.

\bibitem{Pouttu_Conference_EUCNC_2017}
A.~{Pouttu}, J.~{Haapola}, P.~{Ahokangas}, Y.~{Xu},
  M.~{Kopsakangas-Savolainen}, E.~{Porras}, J.~{Matamoros}, C.~{Kalalas},
  J.~{Alonso-Zarate}, F.~D. {Gallego}, J.~M. {Martín}, G.~{Deconinck},
  H.~{Almasalma}, S.~{Clayes}, J.~{Wu}, {Meng Cheng}, F.~{Li}, {Zhipeng Zhang},
  D.~{Rivas}, and S.~{Casado}, ``P2p model for distributed energy trading, grid
  control and ict for local smart grids,'' in \emph{Proc. of European
  Conference on Networks and Communications (EuCNC)}, Oulu, Finland, June 2017,
  pp. 1--6.

\bibitem{Mengelkamp_AppliedEnergy_2017}
E.~Mengelkamp, J.~G{\"a}rttner, K.~Rock, S.~Kessler, L.~Orsini, and
  C.~Weinhardt, ``{Designing microgrid energy markets - A case study: The
  Brooklyn Microgrid},'' \emph{Applied Energy}, vol. 210, pp. 870--880, Jan.
  2018.

\bibitem{PowerLedger_2017}
{Power Ledger}, ``Power {Ledger} white paper,'' Power Ledger Pty Ltd,
  Australia, White paper, 2017,
  \url{https://powerledger.io/media/Power-Ledger-Whitepaper-v3.pdf}.

\bibitem{SmartTest_AE_2018}
C.~Zhang, J.~Wu, Y.~Zhou, M.~Cheng, and C.~Long, ``Peer to peer energy trading
  in a microgrid,'' \emph{Applied Energy}, vol. 220, pp. 1--12, Jun. 2018.

\bibitem{Tushar-TSG:2014}
W.~Tushar, J.~A. Zhang, D.~B. Smith, H.~V. Poor, and S.~Thi{\'{e}}baux,
  ``Prioritizing consumers in smart grid: {A} game theoretic approach,''
  \emph{IEEE Trans. Smart Grid}, vol.~5, no.~3, pp. 1429--1438, May 2014.

\bibitem{Colley_Misc_2014}
A.~Colley, ``Half of users abandon smart meter trial,'' itnews website, 30 Jan.
  2014,
  \url{https://www.itnews.com.au/news/half-of-users-abandon-smart-meter-trial-370919}.

\bibitem{Thomas_PES_2018}
L.~Han, T.~Morstyn, and M.~McCulloch, ``Incentivizing prosumer coalitions with
  energy management using cooperative game theory,'' \emph{IEEE Transactions on
  Power Systems}, 2018, (Early access).

\bibitem{Wong_JSAC_July_2012}
C.~{Joe-Wong}, S.~{Sen}, S.~{Ha}, and M.~{Chiang}, ``Optimized day-ahead
  pricing for smart grids with device-specific scheduling flexibility,''
  \emph{IEEE Journal on Selected Areas in Communications}, vol.~30, no.~6, pp.
  1075--1085, July 2012.

\bibitem{Shams_Energy_July_2018}
M.~H. Shams, M.~Shahabi, and M.~E. Khodayar, ``Stochastic day-ahead scheduling
  of multiple energy carrier microgrids with demand response,'' \emph{Energy},
  vol. 155, pp. 326--338, July 2018.

\bibitem{Baringo_TPWRS_May_2019}
A.~{Baringo}, L.~{Baringo}, and J.~M. {Arroyo}, ``Day-ahead self-scheduling of
  a virtual power plant in energy and reserve electricity markets under
  uncertainty,'' \emph{IEEE Transactions on Power Systems}, vol.~34, no.~3, pp.
  1881--1894, May 2019.

\bibitem{BoChai_TSG_2014}
B.~Chai, J.~Chen, Z.~Yang, and Y.~Zhang, ``Demand response management with
  multiple utility companies: A two-level game approach,'' \emph{IEEE
  Transactions on Smart Grid}, vol.~5, no.~2, pp. 722--731, Mar. 2014.

\bibitem{Maharjan_TSG_Mar_2013}
S.~{Maharjan}, Q.~{Zhu}, Y.~{Zhang}, S.~{Gjessing}, and T.~{Basar},
  ``Dependable demand response management in the smart grid: A stackelberg game
  approach,'' \emph{IEEE Transactions on Smart Grid}, vol.~4, no.~1, pp.
  120--132, Mar. 2013.

\bibitem{Tushar_TSG_Press_2019}
\BIBentryALTinterwordspacing
W.~Tushar, T.~K. Saha, C.~Yuen, T.~Morstyn, N.-A. Masood, H.~V. Poor, and
  R.~Bean, ``Grid influenced peer-to-peer energy trading,'' \emph{IEEE
  Transactions on Smart Grid}, 2019, in press. [Online]. Available:
  \url{https://doi.org/10.1109/TSG.2019.2937981}
\BIBentrySTDinterwordspacing

\bibitem{Tushar_TSG_2_2016}
W.~{Tushar}, B.~{Chai}, C.~{Yuen}, S.~{Huang}, D.~B. {Smith}, H.~V. {Poor}, and
  Z.~{Yang}, ``Energy storage sharing in smart grid: A modified auction-based
  approach,'' \emph{IEEE Transactions on Smart Grid}, vol.~7, no.~3, pp.
  1462--1475, May 2016.

\bibitem{Morstyn_PWRS_2018}
\BIBentryALTinterwordspacing
T.~{Morstyn} and M.~{McCulloch}, ``Multi-class energy management for
  peer-to-peer energy trading driven by prosumer preferences,'' \emph{IEEE
  Transactions on Power Systems}, 2018. [Online]. Available:
  \url{10.1109/TPWRS.2018.2834472}
\BIBentrySTDinterwordspacing

\bibitem{Tushar_Access_Oct_2018}
W.~Tushar, T.~K. Saha, C.~Yuen, P.~Liddell, R.~Bean, and H.~V. Poor., ``{Peer
  to peer energy trading with sustainable user participation: A game theoretic
  approach},'' \emph{IEEE Access}, vol.~6, pp. 62\,932--62\,943, Oct. 2018.

\bibitem{Peterson_JPS_Apr_2018}
S.~B. Peterson, J.~F. Whitacre, and J.~Apt, ``The economics of using plug-in
  hybrid electric vehicle battery packs for grid storage,'' \emph{Journal of
  Power Sources}, vol. 195, no.~8, pp. 2377--2384, Apr. 2010.

\bibitem{Queensland_Electricity_Bill_2019}
\BIBentryALTinterwordspacing
B.~O'Neill, ``{Understanding Queensland electricity tariffs},'' CANSTAR BLUE
  website, July 02 2019. [Online]. Available:
  \url{https://www.canstarblue.com.au/electricity/understanding-queensland-electricity-tariffs/}
\BIBentrySTDinterwordspacing

\bibitem{Queensland_FIT_Bill_2019}
\BIBentryALTinterwordspacing
------, ``{A guide to solar power in Queensland},'' CANSTAR BLUE website,
  September 02 2019. [Online]. Available:
  \url{https://www.canstarblue.com.au/solar-power/guide-solar-power-queensland/}
\BIBentrySTDinterwordspacing

\bibitem{Tushar_TSG_2017}
W.~Tushar, C.~Yuen, D.~B. Smith, and H.~V. Poor, ``Price discrimination for
  energy trading in smart grid: {A} game theoretic approach,'' \emph{IEEE
  Transactions on Smart Grid}, vol.~8, no.~4, pp. 1790--1801, July 2017.

\bibitem{Wayes-J-TSG:2012}
W.~Tushar, W.~Saad, H.~V. Poor, and D.~B. Smith, ``Economics of electric
  vehicle charging: {A} game theoretic approach,'' \emph{IEEE Transactions on
  Smart Grid}, vol.~3, no.~4, pp. 1767--1778, Dec 2012.

\bibitem{Samadi_SmartGridComm_2010}
P.~Samadi, A.~H. Mohsenian-Rad, R.~Schober, V.~W.~S. Wong, and J.~Jatskevich,
  ``Optimal real-time pricing algorithm based on utility maximization for smart
  grid,'' in \emph{IEEE International Conference on Smart Grid Communications},
  Gaithersburg, MD, Oct. 2010, pp. 415--420.

\bibitem{Ma_TCST_2015}
Z.~Ma, S.~Zou, and X.~Liu, ``A distributed charging coordination for
  large-scale plug-in electric vehicles considering battery degradation cost,''
  \emph{IEEE Transactions on Control Systems Technology}, vol.~23, no.~5, pp.
  2044--2052, Sept 2015.

\bibitem{Fahrioglu_GT_1999}
M.~{Fahrioglu} and F.~L. {Alvarado}, ``Designing cost effective demand
  management contracts using game theory,'' in \emph{IEEE Power Engineering
  Society Winter Meeting (Cat. No.99CH36233)}, vol.~1, New York, NY, Jan. 1999,
  pp. 427--432 vol.1.

\bibitem{ChaiBo-TSG:2014}
B.~Chai, J.~Chen, Z.~Yang, and Y.~Zhang, ``Demand response management with
  multiple utility companies: {A} two-level game approach,'' \emph{IEEE Trans.
  Smart Grid}, vol.~5, no.~2, pp. 722--731, March 2014.

\bibitem{Samadi_TSG_Sept_2012}
P.~{Samadi}, H.~{Mohsenian-Rad}, R.~{Schober}, and V.~W.~S. {Wong}, ``Advanced
  demand side management for the future smart grid using mechanism design,''
  \emph{IEEE Transactions on Smart Grid}, vol.~3, no.~3, pp. 1170--1180, Sept.
  2012.

\bibitem{Krzysztof_GTR_2009}
K.~R. Apt and A.~Witzel, ``{A generic approach to coalition formation},''
  \emph{International Game Theory Review}, vol.~11, no.~3, pp. 347--367, Sep.
  2009.

\bibitem{Long_Conf_2017}
C.~Long, J.~Wu, C.~Zhang, L.~Thomas, M.~Cheng, and N.~Jenkins, ``{Peer-to-peer
  energy trading in a community microgrid},'' in \emph{Proc. IEEE PES General
  Meeting}, Chicago, IL, July 2017, pp. 1--5.

\bibitem{RongYu_ITJ:2014}
R.~Yu, J.~Ding, W.~Zhong, Y.~Liu, and S.~Xie, ``{PHEV charging and discharging
  cooperation in V2G networks: A coalition game approach},'' \emph{IEEE
  Internet of Things Journal}, vol.~1, no.~6, pp. 578--589, Dec 2014.

\bibitem{Weule_ABCNews_2018}
\BIBentryALTinterwordspacing
G.~Weule, ``Solar power: {D}oes it make economic sense to buy batteries now or
  should you wait?'' ABC NEWS website, 16 Aug. 2018, accessed on Dec. 09, 2019.
  [Online]. Available:
  \url{https://www.itnews.com.au/news/half-of-users-abandon-smart-meter-trial-370919}
\BIBentrySTDinterwordspacing

\bibitem{Saad-coalition:2009}
W.~Saad, Z.~Han, M.~Debbah, A.~Hj{\o}rungnes, and T.~Ba{\c{s}}ar, ``Coalitional
  game theory for communication networks,'' \emph{IEEE Signal Processing
  Magazine}, vol.~26, no.~5, pp. 77--97, Sept. 2009.

\bibitem{Lee_JSAC_2014}
W.~Lee, L.~Xiang, R.~Schober, and V.~W.~S. Wong, ``Direct electricity trading
  in smart grid: A coalitional game analysis,'' \emph{IEEE Journal on Selected
  Areas in Communications}, vol.~32, no.~7, pp. 1398--1411, July 2014.

\bibitem{MotivationGame_2019}
W.~Tushar, T.~K. Saha, C.~Yuen, T.~Morstyn, M.~D. McCulloch, H.~V. Poor, and
  K.~L. Wood, ``A motivational game-theoretic approach for peer-to-peer energy
  trading in the smart grid,'' \emph{Applied Energy}, vol. 243, pp. 10--20,
  June 2019.

\bibitem{Beebe:1999}
S.~A. Beebe, S.~J. Beebe, and M.~V. Redmond, \emph{Interpersonal Communication:
  Relating to Others}.\hskip 1em plus 0.5em minus 0.4em\relax Boston, MA: Allyn
  \& Bacon., 1999.

\bibitem{Miller:2002}
W.~R. Miller and S.~Rollnick, \emph{Motivational Interviewing: Preparing People
  for Change}.\hskip 1em plus 0.5em minus 0.4em\relax New York, NY: Guilford
  Press, 2002.

\bibitem{Tushar_MPsy_2018}
W.~Tushar, C.~Yuen, W.~T. Li, D.~B. Smith, T.~Saha, and K.~L. Wood,
  ``Motivational psychology driven ac management scheme: A responsive design
  approach,'' \emph{IEEE Transactions on Computational Social Systems}, vol.~5,
  no.~1, pp. 289--301, Mar. 2018.

\bibitem{Shipworth:2000}
M.~Shipworth, ``{Motivating home energy action - A handbook of what works},''
  Australian Greenhouse Office, Australia, Report, Apr. 2000.

\bibitem{Petty_1986}
R.~E. Petty and J.~T. Cacioppo, \emph{{Communication and persuasion: Central
  and peripheral routes to attitude change}}.\hskip 1em plus 0.5em minus
  0.4em\relax New York, NY: Springer-Verlag, 1986.

\bibitem{Hockenbury_2003}
D.~H. Hockenbury and S.~E. Hockenbury, \emph{Psychology}.\hskip 1em plus 0.5em
  minus 0.4em\relax Gordonsville, VA: Worth Publishers, 2003.

\bibitem{Wenzl_JPS_2005}
H.~Wenzl, I.~Baring-Gould, R.~Kaiser, B.~Y. Liaw, P.~Lundsager, J.~Manwell,
  A.~Ruddell, and V.~Svoboda, ``Life prediction of batteries for selecting the
  technically most suitable and cost effective battery,'' \emph{Journal of
  Power Sources}, vol. 144, no.~2, pp. 373--384, June 2005.

\bibitem{2018-EM}
\BIBentryALTinterwordspacing
{ENERGY MATTERS}, ``Australia on the cusp of an energy storage boom,'' website,
  February 15 2018. [Online]. Available:
  \url{https://www.energymatters.com.au/renewable-news/australia-energy-storage-boom/}
\BIBentrySTDinterwordspacing

\bibitem{Walid_TWC_Sept_2009}
W.~Saad, Z.~Han, M.~Debbah, and A.~Hj{\o}rungnes, ``A distributed coalition
  formation framework for fair user cooperation in wireless networks,''
  \emph{IEEE Transactions on Wireless Communications}, vol.~8, no.~9, pp.
  4580--4593, Sept. 2009.

\end{thebibliography}

\end{document}